\documentclass{article}[12pt]
\usepackage[a4paper]{geometry}
\usepackage[english]{babel}
\usepackage[utf8]{inputenc}
\usepackage{amsmath,amssymb,natbib,amsthm}
\usepackage{etoolbox}

\providecommand{\keywords}[1]
{
  \small	
  \textbf{\textit{Keywords:}} #1
}

\newtheorem{theorem}{Theorem}

\newtheorem{proposition}{Proposition}
\newtheorem{assumption}{Assumption}
\newtheorem{example}{Example}

\newtheorem{lemma}{Lemma}

\usepackage{url}

\usepackage[colorinlistoftodos]{todonotes}
\title{Incremental Causal Effects%
}
\date{}

\author{Dominik Rothenh\"ausler and Bin Yu\thanks{Dominik Rothenh\"ausler (rdominik@stanford.edu) is Assistant Professor of Statistics at Stanford University. Bin Yu (binyu@berkeley.edu) is Professor at the Department of Statistics and Department of Electrical Engineering and Computer Sciences at the University of California, Berkeley. Partial supports are gratefully acknowledged from ARO grant W911NF1710005, ONR grant N00014-17-1-2176, NSF grants DMS-1613002 and IIS 1741340, and the Center for Science of Information (CSoI), a US NSF Science and Technology Center, under grant agreement CCF-0939370. BY is a Chan Zuckerberg Biohub investigator. Most of the work of DR was carried out while he was a postdoc with the Yu group in the Department of Statistics at UC Berkeley.
}}

\begin{document}
\maketitle

\begin{abstract}
Causal evidence is needed to act and it is often enough for the evidence to point towards a direction of the effect of an action. For example, policymakers might be interested in estimating the effect of slightly increasing taxes on private spending across the whole population.
We study identifiability and estimation of causal effects, where a continuous treatment is slightly shifted across the whole population (termed average partial effect or incremental causal effect).  We show that incremental effects are identified under local ignorability and local overlap assumptions, where exchangeability and positivity only hold in a neighborhood of units.  Average treatment effects are not identified under these assumptions. In this case, and under a smoothness condition, the incremental effect can be estimated via the average derivative.
  Moreover, we prove that in certain finite-sample observational settings,
estimating the incremental effect is easier than estimating the average treatment effect in terms of asymptotic variance. %
 For high-dimensional settings, we develop a simple feature transformation that allows for doubly-robust estimation and inference of incremental causal effects.   Finally, we compare the behaviour of estimators of the incremental treatment effect and average treatment effect in  experiments including data-inspired simulations.  %
\end{abstract}

\vspace{10pt}

\section{Introduction}

Estimation of treatment effects has a long history in many disciplines and is of central interest in many data science (both scientific and business) endeavours.   Often, data from a randomized experiment is not available and one has to resort to observational data. %
In this situation, practitioners usually mitigate the problem using regression adjustment, matching, inverse probability weighting or instrumental variables regression.

Usually, the overall goal in causal inference is to estimate the effect of intervening on a treatment $T$ on an outcome $Y$. Average treatment effect and conditional average treatment effect are often the quantities of interest in causal inference,
where a treatment variable $T$ is set to the same level across a certain population. In the Neyman-Rubin model, if $Y(t)$ denotes the potential outcome of a unit with treatment assignment $t$  \citep{rubin1974estimating,splawa1990application}, this can be expressed as
\begin{equation*}
  \mathbb{E}[Y(t')] - \mathbb{E}[ Y(t)], \text{ for some fixed $t,t' \in \mathbb{R}$}, 
\end{equation*}
where the expectation is taken over a superpopulation. 

In causal inference, researchers often want to gather evidence to act. However, some actions and causal queries do not correspond to (conditional) average treatment effects. For example, politicians might be interested in estimating the effect slightly increasing taxes on private spending across the whole population. Ride-sharing companies might want to estimate the effect of slightly increasing pay on availability of drivers. Public health experts might be interested in estimating the causal effect of decreasing meat consumption across the population on life expectancy. In general, average treatment effects $\mathbb{E}[Y(t)] - \mathbb{E}[Y(t')]$ or, more generally, the treatment curve $t \mapsto \mathbb{E}[Y(t)]$ do not provide sufficient information to answer domain questions about incremental interventions.
To be more precise, for continuous treatment one may want to estimate the effect of shifting a pre-interventional (potentially random) assignment $T$ by an infinitesimal amount across a certain population.  %
Mathematically, this can be expressed as
\begin{equation}\label{eq:14}
 \mathbb{E}[Y(T+\delta)] - \mathbb{E}[Y(T)] 
\end{equation}
for deterministic $\delta >0 $ close to zero, where the expectation is taken over a superpopulation of units.  %

This notion of intervention is much less used in parts of the causal inference community, but has received attention in the econometric literature. %
The effect in equation~(\ref{eq:14}), divided by $\delta$, corresponds to average partial effects in the econometrics literature, which refer to the effect of an infinitesimal shift in structural equation models \citep{powell1989semiparametric,cameron2005microeconometrics,wooldridge2005unobserved}.
More precisely, in the literature sometimes ``average partial effect'' refers to the causal estimand, sometimes it refers to the functional $\mathbb{E}[\partial_{t} \mathbb{E}[Y|X,T]]$. %
To distinguish the causal estimand from the functional $\mathbb{E}[\partial_{t} \mathbb{E}[Y|X,T]]$, following \citet{kennedy2018nonparametric}, we refer to the estimand in equation~\eqref{eq:14} as an \emph{incremental} treatment effect.

Practitioners have to decide how to specifically formulate a domain question and which notion of intervention to employ to answer the domain question.
Of course, these notions  correspond to different domain questions, and the domain question is one of the most important factors for the choice of intervention notion. 
We would argue that issues of identification, robustness to confounding and difficulty of the estimation task (asymptotic error) should also play important roles in the choice of intervention notion. %

As mentioned previously, treatment effects are often estimated under the assumptions of weak ignorability and the overlap condition. However, these assumptions can be unrealistic or hard to justify if the data is observational. We introduce a local ignorability assumption and a local overlap condition, which are weaker than their more general counterparts. Roughly speaking, the local ignorability assumption states that potential outcomes are independent of the current treatment assignment in a neighborhood of observations. It will turn out that incremental treatment effects are identifiable
under these assumptions, while average treatment effects are not. To the best of our knowledge, tailor-made assumptions to identify incremental treatment effects for continuous treatments have not been defined previously. 

In situations where the practitioner suspects some latent confounding, it would be sensible to employ a notion which is the least sensitive to confounding, if possible. We give a prevalence-sensitive bound that shows that estimation of incremental affect can be relatively stable if small parts of the population are confounded. %
If the signal-to-noise ratio is low, then estimation of average treatment effects might be highly variable and thus uninformative. We will discuss situations in which incremental treatment effects can be estimated with lower asymptotic error than average treatment effects. This includes situations where the signal to noise ratio is low. Thus, estimating incremental treatment effects is potentially more informative than estimating average treatment effects if incremental treatment effects help solve the domain problem.

Causal inference from observational data is known to be challenging and prone to mistakes \citep{rosenbaum2002}, with sometimes devastating effects for human lives. Revisiting the conceptual foundations of the field may help us distinguish situations in which some notions of interventions can be estimated more reliably than others and are thus more informative. In this work, we investigate advantages and disadvantages of incremental treatment effects compared to average treatment effects under the local ignorability  assumption and overlap condition. %
We hope that our work aids practical decisions on choice of intervention notions to answer a specific domain question.

\subsection{Related work}\label{sec:relatedwork}

For discrete treatments, \citet{kennedy2018nonparametric} defines incremental propensity score interventions by multiplying the odds of receiving treatment. The author shows that the overlap assumption is not necessary for identifying incremental propensity score interventions and develops an efficiency theory with corresponding nonparametric estimators. More generally, stochastic interventions have been studied, where treatment assignment is be modified in some pre-specified way. \citet{munoz2012population} and \citet{haneuse2013estimation} study identification and doubly-robust estimation of such treatment effects. Incremental causal effects can be seen as a limit of stochastic interventions, with the change in treatment assignment going to zero.

Asymptotic equivalence of several estimators of the average partial effect in parametric settings has been shown in \citet{stoker1991equivalence}. To the best of our knowledge, semiparametric estimation of average derivatives has first been discussed in \citet{powell1989semiparametric}. More recently, \citet{chernozhukov2018double} and \citet{hirshberg2017augmented} derived semiparametrically efficient estimators of linear functionals of the conditional expectation function.
In \citet{hirshberg2017augmented}, this is achieved under Donsker assumptions, while \citet{chernozhukov2018double} employ a sparsity assumption on either the approximation of the Riesz representer or the approximation to the regression function. %

\citet{hirano2004propensity} generalize the unconfoundedness assumption to continuous treatments. This allows to identify the dose-response function using a generalized propensity score. Under the assumption of our framework, the dose-response function is generally not identifiable.

When there exists a binary instrumental variable, under non-compliance average treatment effects are usually not identifiable. In \cite{angrist1996identification}, the authors show that under a monotonicity assumption, interventions on the subgroup of compliers, the so-called local average treatment effect is still identifiable. In fact, estimating the effect of this ``weaker'' notion of treatment effect in cases where the average treatment effect is not identifiable is increasingly popular in certain parts of the causal inference community \citep{imbens2010better}.

In models based on structure equations, incremental interventions can be seen as a special case of ``parametric'' or ``imperfect'' interventions \citep{Eberhardt2007,korb2004,Tian2001,pearl2014external}. \citet{korb2004}  argue that naively estimating (deterministic) causal effects using Bayesian networks can be misleading and discuss different types of indeterministic interventions. In the context of structure learning, \citet{Eberhardt2007} have shown that causal systems of $N$ variables can be identifiable using one parametric intervention.  \citet{Tian2001} do structure learning based on local mechanism changes, which are a more general notion of intervention than incremental effects. %

\subsection{Our contribution}

We show that the effect of incremental interventions is identifiable under a local ignorability assumption and a local overlap condition, which can be seen as relaxed versions of their more general counterparts. To our knowledge, this is the first time that one introduces taylor-made assumptions for estimation of incremental treatment effects in the potential outcome framework. %

In regression settings, if the distribution of the treatment variable given the covariates is Gaussian, %
using nonlinear basis expansions, we show that subpopulation incremental causal effects can be estimated efficiently with often lower asymptotic error than average treatment effects.
Furthermore, we show that the estimation of incremental causal effect in a regression %
setting in the population case is relatively little affected if a small portion of the population is confounded. %
We propose a feature transformation ``incremental effect orthogonalization'' that facilitates estimation and inference of high-dimensional incremental causal effects. Our method is based on a feature transformation in the first step and running a de-sparsified lasso on the transformed data.
The de-sparsifying technique is similar to the ones developed for high-dimensional linear regression \citep{zhang2014confidence,javanmard2014confidence,van2014asymptotically,belloni2014inference} but due to the randomness in treatment assignment the asymptotic variance formula differs from existing approaches for inference in high-dimensional regression models. We modify a robust version of the desparsifying approach which was introduced in \cite{buhlmann2015high}. Our novel technique has a double-robustness property, in the sense that if one of two models is well-specified, we obtain asymptotically valid estimation and inference. The main advantage compared to existing approaches for the estimation of incremental treatment effects is that after a simple feature transformation, off-the-shelf software for estimation in high-dimensional linear models can be used.  %

To substantiate the claims of our theoretical results, we compare the behaviour of estimators of the incremental treatment effect and average treatment effect on simulated data. We also cover simulations settings where the assumptions of some of our theoretical results are violated and discover that the conclusions are fairly robust.%

The rest of the paper is organized as follows. In Section~\ref{sec:motivation-setup} we introduce the model class. In Section~\ref{sec:ident-incr-caus}, we discuss identification of the incremental treatment effect under a local ignorability and local overlap condition. In Section~\ref{sec:prop-our-estim}, we discuss models in which estimating incremental treatment effects can be done with lower asymptotic error than estimating average treatment effects and derive bounds under confounding. Furthermore, we introduce a feature transformation that facilitates doubly robust estimation and inference in high-dimensional settings. Finally, in Section~\ref{sec:sim} we validate our theoretical results on simulated data, including set-ups based on real-world data. %

\section{Motivation and Setup}\label{sec:motivation-setup}

We are interested in the causal effect of a continuous treatment variable $T$ on an outcome $Y$ in the presence of some covariates $Z$. We use the potential outcome framework \citep{rubin1974estimating,splawa1990application} and assume a super population or distribution $\mathbb{P}$ of $(Y(t)_{t \in \mathbb{R}},T,X)$ from which $n$ independent draws $(Y_{i}(T_{i}),T_{i},Z_{i})$ are given, where $Y_{i}(t)$ is the potential outcome of $Y$ under treatment or dose $T=t$.  %
Without any assumptions or adjustment, an observed association between $T$ and $Y$ might simply be due to some confounding variable that affects both the treatment and the outcome. A commonly made assumption in such a setting is weak ignorability \citep{rosenbaum1983assessing}, which states that the treatment assignment is independent of the outcome, conditional on some covariates $Z$. Formally, this assumption is often written as
\begin{equation}\label{eq:1}
    \{Y(t), t \in \mathcal{T} \} \perp T | Z,
\end{equation}
where $Y(t)$ is defined as the potential outcome of $Y$ under treatment $T=t$ and $\mathcal{T}$ is the set of treatment levels. To avoid issues of measurability, we assume that $Y(t)$ is continuous.
In addition, it is common to assume that the overlap condition holds, which can be written as
\begin{equation}\label{eq:globaloverlap}
     p(t|z) > 0 \text{ for all $t,z$, }
\end{equation}
where $p(t|z)$ is the conditional density of $T$ given $Z$ and $p(z)$ is the density of $Z$. If both weak ignorability and the overlap condition holds, the average treatment effect $\mathbb{E}[Y(t)] - \mathbb{E}[Y(t')]$ for some choice of $t,t' \in \mathcal{T}$ can be estimated via regression, matching, inverse probability weighting or combinations of these methods, see for example \citet{rosenbaum2002}. In particular,
\begin{equation*}
  \mathbb{E}[Y(t)] - \mathbb{E}[Y(t')] = \mathbb{E}[\mathbb{E}[Y|Z,T=t]]-\mathbb{E}[\mathbb{E}[Y|Z,T=t']],
\end{equation*}
where on the right-hand side the outer expectation is taken over the distribution of $Z$. %

In the following we will show that incremental causal effects are identifiable in scenarios where the ignorability assumptions holds locally in subgroups of patients with similar treatments but not necessarily across all patients. In addition, we show that the overlap condition can be weakened as well.
Average treatment effects and the dose-response function are generally not identifiable under these assumptions. %

Throughout the paper we assume that the Stable Unit Treatment Value Assumption \citep{rubin1980discussion,hernan2010causal} holds, which says that the potential outcomes are well-defined and that there is no interference between units, i.e.\ that the potential outcome of one individual is only a function of the treatment assignment to this individual and not of the others. Formally, this assumption can be expressed as $Y_{i}(\mathbf{T})=Y_{i}(T_{i})$ for $\mathbf{T} = (T_{1},\ldots,T_{n})$. 
\subsection{Local ignorability}\label{sec:local-ignorability}
In practice there are cases where the weak ignorability assumption in equation~\eqref{eq:1} is unrealistic. Often, it is unknown for which set of covariates $Z$ the weak ignorability assumption holds. For example, a doctor may want to estimate the effect of changing the dose of some medication on health outcomes in patients. Even if the set $Z$ is known, it may be costly or impractical to collect the full set of covariates $Z$ for every patient. The doctor may only be able to collect a smaller set of covariates $X \in \mathbb{R}^{d}$. If $T$ describes the dose of a medication, the patients receiving high doses of the medication may be a completely different population than the patients receiving low doses of the medication, even conditional on the observed covariates $X$. If $X = \text{``severity of symptoms''}$ 
and there is a patient that has severe symptoms and receives a a very low dose, this might be due to the doctor making an exceptional decision due to an exceptional circumstance that is not encoded in the data set. This exceptional circumstance might also affect the outcome.
It could also be that the patient willingly decides to take only a lower dose than usual. Consequentially, this patient may make other decisions that affect $Y(t)$ that are not encoded in $X$. In the following, we introduce a localized version of the ignorability assumption which allows for some unobserved heterogeneity.
\begin{assumption}[Local ignorability]\label{ass:local-ignorability}
Let equation~\eqref{eq:1} hold for some set of covariates $Z$. Let $Z = (X,H)$, where $X$ is observed and $H$ is unobserved.

Assume that for all $(t_{0},x_{0})$ with $p(t_{0},x_{0}) > 0$ for all $(t,x)$ close to $(t_{0},x_{0})$,
  \begin{equation*}
    p(t|x) = p(t|x,h) %
  \end{equation*}
for all $h$ with $p(h|t,x) >0$  or
  \begin{equation*}
    \mathbb{E}[Y|T=t,X=x] = \mathbb{E}[Y|T=t,X=x,H=h] %
  \end{equation*}
for all $h$ with $p(h|t,x) >0$. We also allow for the special case where $X=Z$ (i.e., no confounding) and the special case where $H=Z$, i.e. where there are no observed covariates. In the former case, we assume equation~\eqref{eq:1}. If $H=Z$, we demand that for all $t_{0}$ with $p(t_{0})>0$, for $t$ close to $t_{0}$ either $p(t) = p(t|h)$ or $\mathbb{E}[Y|T=t] = \mathbb{E}[Y|T=t,H=h]$ for all $h$ with $p(h|t) >0$.
\end{assumption}
Roughly speaking, conditionally on $X=x$ and $T=t$, we assume that either the treatment assignment probabilities are locally constant in $h$ or the expected outcome is locally constant in $h$. This means we assume that locally the units are either homogeneous in their treatment assignment or homogeneous in their outcome, but not necessarily both.

This can happen, for example, if a doctor has discrete groups of patients that he treats differently. Then patients with similar treatments might be exchangeable, but patients with very different treatments might be confounded. Figure~\ref{fig:pic} contains a visual example where local ignorability holds and the standard ignorability assumption might be violated. A concrete example with the data generating process can be found below. %

\begin{example}\label{ex:local}

  Consider the following data generating process:
  \begin{align*}
    \epsilon &\sim \mathcal{N}(0,1)\\
    H &\sim \mathcal{N}(0,1) \\
   T &= \begin{cases}
     T = - |\epsilon|%
     &\text{ if $H \le 0$},\\
     T = |\epsilon| &\text{ if $H \ge 0$}
   \end{cases}\\
   Y(t) &= 2 t + (t-1) 1_{t \ge 1} \cdot H. 
  \end{align*}
The ignorability assumption is not satisfied as $\{Y(t)\}_{t \in \mathbb{R}}$ is not independent of $T$. Roughly speaking, subjects with $t > 0$ have the same propensity score as treatment assignment $T$ is independent of $H$ conditioned on $T \ge 0$. Subjects with $t < 1$ have the same outcomes in distribution as the outcome function does not depend on $H$. Thus, the ignorability assumption holds locally. 
 We now make this intuition more precise.

 For all $t < 1$ and all $h$,
\begin{equation*}
  \mathbb{E}[Y|T=t] = 2t = \mathbb{E}[Y|T=t,H=h].
\end{equation*}
For $t>0$, $p(t,h) >0$ only if $h > 0$. For $h > 0$,
\begin{equation*}
  p(t|h) = p(t).
\end{equation*}
Thus, the local ignorability assumption is satisfied.
\end{example}
\begin{figure*}
\begin{center}
 \includegraphics[scale=.5]{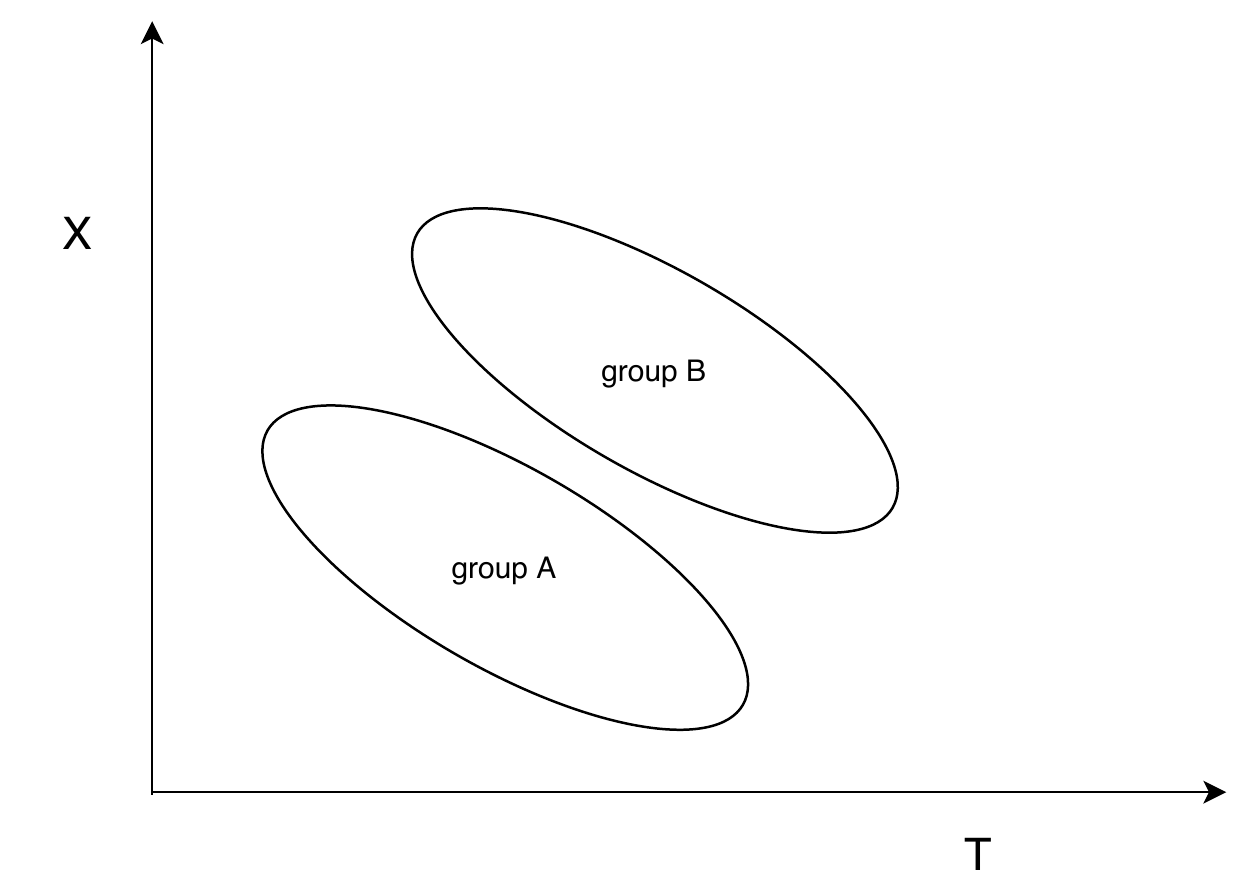}
\end{center}
\caption{In this example, assume that the weak ignorability assumption holds for each of the patient groups A and B. In general, conditionally on $X$, the ignorability assumption is violated for the total patient population. The overlap assumption is violated as well. However, the local ignorability assumption can still hold. For a concrete example where ignorability holds locally but not globally and the groups are not separated, see Example~\ref{ex:local}.}\label{fig:pic}
\end{figure*}
Thus, the local ignorability assumption allows for some unobserved heterogeneity and is strictly weaker than the standard ignorability assumption.  %

\subsection{Local overlap}

Now let us turn to weakening the overlap condition. %
If the overlap condition is not satisfied, parts of the population are observed with probability zero for some treatment $T=t$. Estimating a causal effect for this part of the population amounts to extrapolating from other parts of the population, which is naturally prone to errors. In practice, the overlap condition is often not satisfied. In the example discussed above, ethical considerations, among others, might prevent doctors to give a very low dose to patients with severe symptoms and a very high dose to patients with minor symptoms. If $X = \text{``severity of symptoms''}$, then by conditioning on a large value $X=x$ we will have only very few patients with low doses, or no patients at all. This makes it exceedingly difficult to estimate the effect of giving a low dose to this group of patients. In finite samples, estimation of the average treatment effects is difficult if there exists regions where assignment variables $t$ has low density and the other treatment assignment $t'$ has high density. This issue can be exacerbated if the covariate vector $X$ is high-dimensional and the data is observational \citep{d2017overlap}. Roughly speaking, due to the curse of dimensionality, in high-dimensions there will often be regions where one of the treatments has low density.  A similar problem appears in practice when trying to match subjects on many covariates. The more covariates we have, the harder it is to find pairs of subjects that are similar in all covariates.

High-dimensional covariates are potentially also challenging for the estimation of incremental effects. However, for the estimation of incremental effects,
we can relax the overlap condition. In the following we assume that the density functions $p(x)$, $p(h|x)$, $p(t|x)$ and $p(t|x,h)$ exist.

\begin{assumption}[Local overlap]\label{ass:local-overlap}
  Assume that %
  $p(t|x,h)$ and $p(t|x)$ are continuous in $t$.  
\end{assumption}

Roughly speaking, we assume that if there is a patient with attributes $X=x$ and $H=h$ that gets treatment $T=t$ then the probability of another patient with the same attributes having a slightly different treatment is nonzero. An example is given in Figure~\ref{fig:local_overlap}.
\begin{figure*}
\begin{center}
 \includegraphics[scale=.5]{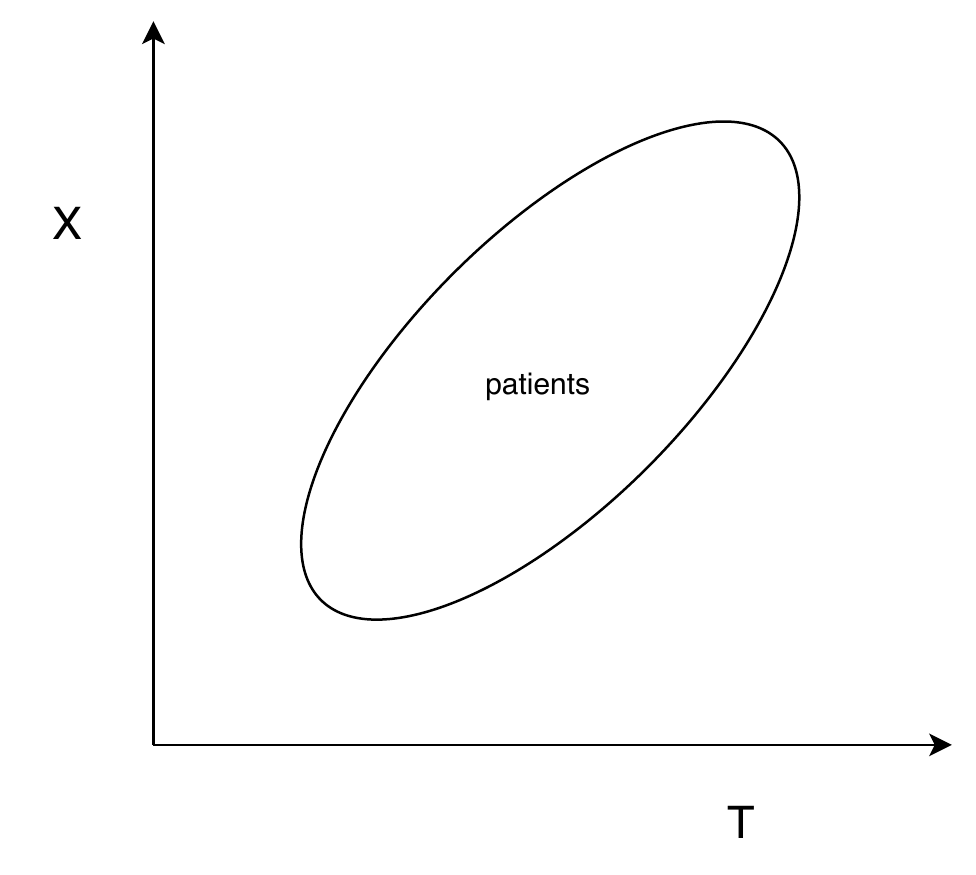}
\end{center}
\caption{Let the ellipsoid denote the region of positive density, $\{(x,t): p(x,t)>0\}$. In this example, the overlap condition is violated as there are no patients that have strong symptoms $X=x$ but get a low dose $T=t$. On the other hand, the local overlap condition holds if the densities are continuous.}\label{fig:local_overlap}
\end{figure*}

\subsection{Regularity assumption} 

In addition, we make an assumption to guarantee that the conditional expectations used in this paper exist. In practice, these assumptions can be thought of as putting a smoothness condition on the potential outcomes, on regression surfaces, and the density. %

\begin{assumption}[Regularity]\label{ass:regul-assumpt}
Assume that the potential outcomes $Y(t)$ are bounded and that the derivative $Y'(t) := \partial_{t} Y(t)$ is continuous and bounded. %
 Furthermore, assume that the conditional densities $ p(t|x)$ and $p(t|x,h)$ are differentiable in $t$ and that $ p(t|x) \rightarrow 0$ and $p(t|x,h) \rightarrow 0$ for all fixed $x$ and $|t| \rightarrow \infty$.
\end{assumption}
Regularity assumptions of this type are necessary to different expressions of incremental causal effects well-defined.

\section{Identification of incremental effects}\label{sec:ident-incr-caus}
Note that if we only have the local ignorability assumption or the local overlap condition, it is generally not possible to consistently estimate (or identify) the average treatment effect
\begin{equation*}
   \mathbb{E}[Y(t)] - \mathbb{E}[Y(t')] \text{ for some $t,t' \in \mathbb{R}$.}
\end{equation*}
However, we will now show that the effect of shifting $t$ by a small amount around the current treatment can still be identifiable. In the following, we assume that Assumption~\ref{ass:regul-assumpt} holds, i.e.\ we make the assumption that $Y(t)$ is bounded and continuously differentiable in $t$ with bounded derivative. %
The short proof of the following result can be found in  Section~\ref{sec:proofmain} in the Appendix. 
\begin{proposition}\label{prop:ident}
  If Assumption~\ref{ass:local-ignorability} (local ignorability), Assumption~\ref{ass:local-overlap} (local overlap) and Assumption~\ref{ass:regul-assumpt} (regularity) are satisfied, %
  then for almost all $(x,t)$ with $p(x,t)>0$,
\begin{align}\label{eq:27}
\begin{split}
     \mathbb{E}[ Y'(t) |T=t,X=x]    &= \partial_{t} \mathbb{E}[Y|T=t,X=x]
\end{split}
\end{align}
\end{proposition}
This result states that if the observations are locally either homogeneous in their response  $ \mathbb{E}[Y|T=t,X=x] = \mathbb{E}[Y|T=t,X=x,H=h]$ or homogeneous in their treatment assignment $p(t|x) = p(t|x,h)$, then the causal effect is identified from the observational distribution via the formula given in the theorem. This is a strictly weaker assumption than ignorability as discussed in Section~\ref{sec:local-ignorability}.

  From a practical perspective, Proposition~\ref{prop:ident} may help guide data collection for causal inference:  First, a researcher can determine a set of covariates $Z$ that is plausible to satisfy the ignorability assumption, but might be impossible to collect in practice. Then, the researcher can aim to collect a set of the covariates $X$ such that for each $x$, the propensity score is close to the oracle propensity score $p(t|x) \approx p(t|x,h)$ or outcome function is close to the oracle outcome function $ \mathbb{E}[Y|T=t,X=x] \approx \mathbb{E}[Y|T=t,X=x,H=h]$. If collecting the data is costly and resources are limited, the researcher can prioritize collecting additional covariates for patients for which the researcher suspects high variability in the outcome $\mathbb{E}[Y|T=t,X=x,H=h]$ or propensity $p(t|x,h)$.

One way to think about this result is that we have a double robustness phenomenon. If for every $(t,x,h)$ with $p(t,x,h) >0$ locally either the observed propensity $p(t|x)$ coincides with the oracle propensity $p(t|x,h)$, or the outcome function $\mathbb{E}[Y|T=t,X=x]$ coincides with the oracle outcome function $\mathbb{E}[Y|T=t,X=x,H=h]$, then the causal effect is identifiable via the functional given in equation~\eqref{eq:27}. However, note that it is distinctly different from standard double robustness in causal inference, as it is \emph{method agnostic}. More specifically, under local ignorability we can consistently estimate the population incremental causal effect $\theta_{sp}$ via regression adjustment, via weighting, or via doubly robust methods.

In practice, the local ignorability assumption may only approximately. Bounds for the bias under violations of the local ignorability assumption can be found in Section~\ref{sec:sens-addit-conf}.

Note that Assumption~\ref{ass:regul-assumpt}, which assumes  that $Y(t)$ is bounded and differentiable with bounded derivative, is made for expositional clarity and can be weakened slightly. %

In the following we have to differentiate between the incremental effect for the superpopulation and the incremental effect for the sample. In the latter case, all statements are conditional on the units that we observe. We will consider estimation and inference for both sample and superpopulation incremental effects. More information for the difference between population causal effects and finite sample causal effects can be found in \cite{imbens2015causal}, Section~1.
More specifically, we will investigate estimation and inference 
for the expected effect of an infinitesimal shift intervention for a finite sample $(y_{i},t_i,x_i)$, $i=1,\ldots,n$,
\begin{equation}\label{eq:globsa}
  \theta_\text{fs} :=   \frac{1}{n}\sum_{i=1}^n \mathbb{E}[ Y'(t_{i})|T=t_{i},X=x_i],
\end{equation}
where the expectation is taken with respect to the potential outcomes $Y$ given $X=x$ and $T=t$. Similar conditional average treatment effects have been used in \citet{armstrong2018finite}. %
We allow for randomness in $Y$ after conditioning on $X=x$ and $T=t$. This can account for within-subject variability of treatment effects. 
For example, for a given subject we allow the treatment effect to vary across time. In words, $\theta_{\text{fs}}$ corresponds to the expected effect of an incremental intervention conditionally on $(t_{i},x_{i})_{i=1,\ldots,n}$. We condition on the $(t_{i})_{i=1,\ldots,n}$ as we aim to estimate the effect of slightly shifting the current treatment assignments $(t_{i})_{i=1,\ldots,n}$. Note that we do not condition on the potential outcomes %
as is sometimes done in the causal inference literature for binary treatments \citep{imbens2015causal}. Conditioning on the outcomes, in addition to conditioning on the $t_{i}$ and $x_{i}$, would result in a deterministic sample and is thus not practical. %

For $n \rightarrow \infty$, under some regularity conditions and a super population model, $\theta_\text{fs}$ will converge to the super population effect, %
\begin{equation}\label{eq:glob}
  \theta_\text{sp} :=  \mathbb{E}[Y'(T)].
\end{equation}
Here, the expectation is taken over $T$, $X$ and $Y$. In words, $\theta_{\text{sp}}$ corresponds to the expected effect of an incremental intervention on the superpopulation. %
We want to emphasize that the interpretation of these causal effects is different from the most common notion of interventions, so-called average treatment effects, surgical interventions or do-interventions. %
The population incremental causal effect %
answers the question: ``How will the average outcome change if we change the treatment of all patients by a small amount, i.e. use treatment assignment $T' = T + \delta$ across all patients'' for some $\delta$ close to zero. Note that the quantities of interest are in general also different from
\begin{equation*}
\mathbb{E}[Y'(t)], 
\end{equation*}
which corresponds to first setting the value of $T$ to $t$ across the whole population and than varying that intervention by a small amount. In the next section we will discuss general properties of estimation and inference of the effect in equation~\eqref{eq:globsa} and 
equation~\eqref{eq:glob}. %
\section{Estimating incremental effects in regression settings}\label{sec:prop-our-estim}

In this section we discuss various aspects of estimation and inference for incremental causal effects for the sample incremental effect $\theta_{\text{fs}}$ and the population incremental effect $\theta_{\text{sp}}$. 
In Section~\ref{sec:smaller} we discuss regression settings for which the error of efficiently estimating the effect of shift interventions is lower than the error of efficiently estimating average treatment effects. This is achieved if the conditional distribution $p(t|x)$ is Gaussian. In Section~\ref{sec:sens-addit-conf} we show that estimation of incremental interventions can be relatively little affected if small parts of the population are confounded. %
In Section~\ref{sec:high-dimens-conf} we discuss a feature transformation that facilitates obtaining confidence statements in high-dimensional scenarios.
The main advantage compared to existing approaches \citep{powell1989semiparametric,hirshberg2017augmented,chernozhukov2018double} is that after a simple feature transformation, off-the-shelf software for estimation and inference in high-dimensional linear models can be used. %

\subsection{Variance comparison}\label{sec:smaller}

In this section we will see that in some standard regression settings, the squared error of estimating the sample average treatment effect is often larger than the squared error of estimating the sample effect of incremental interventions for regression estimators.
We assume that we observe $n$ i.i.d. observations $(x_i,t_i,y_i)$, $i=1,\ldots,n$, of distribution $\mathbb{P}$ that follow the additive noise model
\begin{equation}\label{eq:model}
    Y = f^0(T,X) + \epsilon,
\end{equation}
where $\epsilon$ is independent of $T$ and $X$, $\mathbb{E}[\epsilon]=0$ and $\text{Var}(\epsilon)= \sigma_\epsilon^2 > 0$. We assume that $f^{0}$ is differentiable in $t$, i.e.\ we make a smoothness assumption on how the treatment affects the outcome. 
As an example, $y$ could be the math score of a student in a test, $t$ the study time after class and $x$ pre-treatment covariates such as age and gender. $\epsilon$ could be the influence of some other unmeasured factors that are independent of $X$ and $T$ but have an influence on the math score. In the easiest case, the function $f^{0}$ is linear in its arguments, but often there will be interactions between the pre-treatment covariates and study time after class. If $x = \text{``doing at least two hours of sports per week''}$, then the effect of studying on the outcome $y$ might be stronger if the student does enough sports. %

Under the assumptions of Proposition~\ref{prop:ident}, we have $\theta_\text{fs} = \frac{1}{n} \sum_{i=1}^n \partial_t f^0 (t_i,x_i)$.  We can estimate $\theta_\text{fs}$ via a two-stage procedure in non-parametric %
settings. First, choose differentiable basis functions $b_{1},\ldots, b_{p}$. %
These functions can for example include linear terms, polynomials, radial basis functions or wavelets. Ideally, the choice of basis functions is guided by domain knowledge. In cases where domain knowledge is not available, we recommend using linear and quadratic terms.
First, we solve the least-squares problem 
\begin{equation*}
   \hat \beta = \arg \min_{\beta} \sum_{i=1}^n \left(y_i - \sum_{k=1}^{p}b_{k}(t_{i},x_i) \beta_{k} \right)^2.
\end{equation*}
Secondly, we estimate the average derivative via
$$\hat \theta_{\text{fs}} = \frac{1}{n} \sum_i \sum_{k} \hat \beta_{k} \partial_t  b_{k}(t_{i},x_i).$$
For some $t' \neq t$, we compare the performance of this estimator with a naive estimator
\begin{equation}\label{eq:defhattheta}
    \hat \tau(t,t') = \frac{1}{n}\sum_{i=1}^n \hat f(t,x_i) - \hat f(t',x_i),
\end{equation}
where $\hat{f} = \sum_{k} \hat \beta_{k} b_{k}(t_{i},x_{i})$. %
Defining and estimating finite-sample average treatment effects is challenging in observational settings. Here, following \citet{armstrong2018finite}, we define finite sample average treatment effects in observational settings by conditioning on the covariates: Under weak ignorability and regularity assumptions, $\hat \tau(t,t')$ is a consistent estimator of the sample average treatment effect
\begin{align*}
 & \tau_\text{fs}(t,t')  := \\
 & \frac{1}{n} \sum_{i=1}^n \mathbb{E}[Y(t)|X=x_i] - \mathbb{E}[Y(t')|X=x_i].
\end{align*}
Define $\mathcal{B}$ as the linear span of $b_{1},\ldots,b_{p}$. Write $\textbf{X}_{j;i} = b_{j}(t_{i},x_{i})$.  Here and in the following, to be able to use partial integration we tacitly assume that $b_{k}(t,x) p(t|x) \rightarrow 0$ and $\partial_{t} b_{k}(t,x) p(t|x) \rightarrow 0$ for fixed $x$ and $h$ and $ |t| \rightarrow \infty$.
\begin{theorem}[Asymptotic variance comparison]\label{thm:variance_comp}
 Assume that the data $(y_i,t_i,x_i)$, $i=1,\ldots,n$ are i.i.d. and follow the model in equation~\eqref{eq:model}. Furthermore, assume that $ \partial_t \log p(t,x) \in \mathcal{B}$ and $f^{0} \in \mathcal{B}$. Let the basis $b_{1}(t,x),\ldots,b_{p}(t,x)$ be differentiable %
 and assume that the basis $(b_k(T,X))_{k=1,\ldots,p}$ and its partial derivative $(\partial_{t} b_{k}(T,X))_{k=1,\ldots,p}$ have finite second moments.
If the conditional distribution $p(t|x)$ %
 is Gaussian for all $x$ with constant variance
then for all $t,t' \in \mathbb{R}$,
\begin{equation*}
  \limsup_{n \rightarrow \infty} \frac{\mathbb{E}[(\hat \theta_{\text{fs}} - \theta_{\text{fs}})^{2}|\mathcal{D}_{\text{feat}}]}{\mathbb{E}[(\hat \tau(t,t') - \tau_{\text{fs}}(t,t'))^{2}|\mathcal{D}_{\text{feat}}]} \le \frac{1}{(t-t')^{2}},
\end{equation*}
where $\mathcal{D}_{\text{feat}} = \{(x_{i},t_{i}),i=1,\ldots,n\}$.
\end{theorem}
The theorem implies that the simple plug-in estimator $\hat \theta$ has asymptotically lower error or the same error as $\hat \tau$ if $\mathcal{B}$ is large enough, $p(t|x)$ is Gaussian and $|t-t'| \ge 1$. The difference in asymptotic variance can be drastic, as we will see in the simulation section. Hence, the concept of incremental interventions might be helpful in situations where the signal-to-noise ratio is too low for drawing any conclusions from $\hat \tau$.

The estimators above are both efficient for estimating their respective target quantities within the class of unbiased and linear estimators \citep{van2000asymptotic}. In this sense, the result above shows that in settings where the basis is fixed and $n \rightarrow \infty$, estimating sample incremental causal effects $\theta_{\text{fs}}$ can be easier in terms of optimal asymptotic variance. Thus, if the domain problem is adequately addressed in the incremental causal effect formulation, estimating incremental causal effects might be more informative than estimating average treatment effects. %

We will discuss scenarios in which estimation of average treatment effects results in lower asymptotic variance than estimating incremental treatment effects in Section~\ref{sec:sim}. %

\subsection{Sensitivity to confounding}\label{sec:sens-addit-conf}

In this section we discuss how estimation of incremental causal effects behaves in the super-population case if the local ignorability assumption is slightly violated. Analyzing the sensitivity of causal effects with respect to violations of assumptions plays a central role in observational studies, see for example \citet{cornfield1959smoking, rosenbaum1983assessing, rosenbaum2002}.

Roughly speaking, we will see that estimation of incremental treatment effects is relatively unaffected by confounding $H$ if, conditionally on $X$ and $H$, units have either similar treatment assignment or outcomes across $H$. %
We assume that the ignorability assumption holds, conditionally on the covariates $X$ and some unobserved variable $H$,
\begin{equation*}
  Y(t) \perp T | X,H.
\end{equation*}
Thus, using partial integration and regularity (Assumption~\ref{ass:regul-assumpt}), the superpopulation incremental causal effect is
\begin{align*}
  \theta_{\text{sp}} &= \mathbb{E}[\partial_t Y(T)] \\
  &= \mathbb{E}[\partial_t \mathbb{E}[Y|X,T,H]] \\
  &=  \mathbb{E}[ - \partial_{t} \log p(T|X,H) \cdot Y].
\end{align*}
However, as $H$ is unobserved, for $n \rightarrow \infty$ the procedure described above converges to
\begin{equation}\label{eq:29}
  \theta_{\text{estimated}} = \mathbb{E}[\partial_t \mathbb{E}[Y|X,T]] = \mathbb{E}[ - \partial_{t} \log p(T|X) \cdot Y].
\end{equation}
We aim to understand how far $\theta_{\text{estimated}}$ is from $\theta_{\text{sp}}$, based on the strength of confounding.  The proof of the following result can be found in the Appendix.
\begin{proposition}\label{prop:robustn-addit-conf}
We have
\begin{align}\label{eq:dev}
\begin{split}  
   & | \theta_{\text{sp}} - \theta_{\text{estimated}} | \\
  \le &\mathbb{E} \large[ |\mathbb{E}[Y|T,X] - \mathbb{E}[Y|T,X,H]| \cdot |\partial_{t} \log p(T|X) - \partial_{t} \log p(T|X,H)| \large],
\end{split} 
\end{align}
\end{proposition}
 This result shows either if $\mathbb{E}[Y|T,X] \approx \mathbb{E}[Y|T,X,H]$ or $\partial_{t} \log p(T|X) \approx \partial_{t} \log p(T|X,H)$, then the true causal effect will be close to the estimated causal effect in the population case. Note that on the right hand-side of equation~\eqref{eq:dev}, the deviation is measured in terms of the integrated $\ell_1$ norm. If the local ignorability assumption holds, the right-hand side in equation~\eqref{eq:dev} is zero. Intuitively speaking, this means that the right-hand side of equation~\eqref{eq:dev} measures the strength of confounding, averaged over the population.  This has the consequence that the bound can be informative (and the estimated effect can be close to the true effect) if a small subset of the population is strongly confounded. 
\subsection{Doubly-robust estimation and confidence intervals under sparsity}\label{sec:high-dimens-conf}

In this section we describe a simple procedure to derive asymptotically valid confidence intervals using the lasso. ``Doubly robust'' is meant in the sense that the method yields asymptotically valid confidence intervals if the function class $\mathcal{B}$ either contains $\partial_{t} \log p(t,x)$ or $f^{0}$. The method we describe is based on the method on page~6 in \cite{buhlmann2015high}. However, we will transform the features in a pre-processing step. This transformation will depend on the observations $(t_{i},x_{i})$ and change for every $n$, which adds some complexity. The benefit of this pre-processing step is to reduce the problem of estimating incremental treatment effects to a problem of estimating a single component in a high-dimensional, potentially misspecified linear model. As a result, existing software such as the \texttt{R}-package \texttt{hdi} \citep{hdi} can be used to efficiently estimate incremental treatment effects. Of course, the results below also extend to the low-dimensional case, i.e.\ to the case where the number of features $p$ is fixed and the number of observations $n$ goes to infinity.

As above, consider basis functions $b_{1}(t,x),\ldots,b_{p}(t,x)$ of $\mathcal{B}$ which are potentially non-linear. To make notation simpler without loss of generality we assume that $b_{1}(t,x)=t$.
Assume we have an i.i.d. sample $(y_{i},t_{i},x_{i})$, $i=1,\ldots,n$, and define the feature matrix $\mathbf{X}$ via $\mathbf{X}_{k;i} = b_{k}(t_{i},x_{i})$, the target vector $\mathbf{Y} = (y_{1},\ldots,y_{n})^{\intercal}$ and the transformed feature matrix $\tilde{\mathbf{X}}$ via
\begin{equation*}
  \tilde{\mathbf{X}}_{k;i} =  \begin{cases} t_{i} & \text{ for $k=1$,} \\
    b_{k}(t_{i},x_{i}) - t_{i} \frac{1}{n} \sum_{i=1}^{n} \partial_{t} b_{k}(t_{i},x_{i}) & \text{ for $k>1$}.
    \end{cases} 
  \end{equation*}
   This can be thought of as an orthogonalization step. %
   For $n \rightarrow \infty$ and $k > 1$,
   \begin{equation*}
       \tilde{\textbf{X}}_{k,i} \rightarrow b_k(t_i,x_i) - t_i \mathbb{E}[ \partial_t b_k(t_1,x_1)].
   \end{equation*}
   In addition, using partial integration, it can be shown that
   \begin{equation}
   \mathbb{E}[ \partial_t \log p(t_i | x_i) \cdot (    b_k(t_i,x_i) - t_i \mathbb{E}[ \partial_t b_k(t_1,x_1)]) ] = 0.
   \end{equation}
   As sketched above, for $k>1$ the features $(\tilde{\mathbf{X}}_{k;i})_{i=1,\ldots,n}$ are asymptotically uncorrelated of $ (-\partial_{t} \log p(t_{i},x_{i}))_{i=1,\ldots,n}$. This orthogonality property will turn out to be useful for asymptotic optimality of our proposed procedure as $-\partial_t \log p(t|x)$ is the so-called Riesz representer of the average derivative \citep{chernozhukov2018double}.
  Using implicit or explicit orthogonality properties is common in semiparametric approaches. For general functionals the orthogonalization step can be more involved \citep{chernozhukov2018double}. 
  
  In finite samples, the orthogonalization step can be motivated as follows. Define the transformed functions
\begin{equation*}
  \tilde b_{k} = \begin{cases}
    t & \text{ for } k=1,\\
    b_{k}(t,x) - t \alpha_{k} & \text{ for } k > 1,
    \end{cases}
\end{equation*}
with $\alpha_{k} = \frac{1}{n} \sum_{i=1}^{n} \partial_{t} b_{k}(t_{i},x_{i})$.
Then also $\tilde b_{1},\ldots,\tilde b_{p}$ is a basis of $\mathcal{B}$. Hence, if $f^{0} \in \mathcal{B}$, we can write it as $f^{0} = \sum_{k} \beta_{k}^{0} \tilde b_{k}$.
  And in particular, we have the average derivative
\begin{equation}
\theta_{\text{fs}} = \frac{1}{n} \sum_{i=1}^{n} \partial_{t} f^{0}(t_{i},x_{i}) =   \sum_{k=1}^{p} \beta_{k}^{0} \frac{1}{n} \sum_{i=1}^{n} \partial_{t} \tilde b_{k} = \beta_{1}^{0}.
\end{equation}
In words, the average derivative of the function corresponds to the coefficient in front of $\tilde b_{1}$ in this new basis. The following proposition is a direct result of the observation above.

\begin{proposition}
 Assume that $b_{1},\ldots,b_{p}$ is a basis of the function class $\mathcal{B}$. Assume that $Y = f^{0}(T,X) + \epsilon$ for $\epsilon$ independent of $(T,X)$ and that $f^{0} \in \mathcal{B}$. Let $(y_{i},t_{i},x_{i},\epsilon_{i})_{i=1,\ldots,n}$ be i.i.d. with the same distribution as $(Y,T,X,\epsilon)$ and denote $\mathbb{E}_{\epsilon}$ the expectation with respect to the $\epsilon_{i}$, $i=1,\ldots,n$. Then,
  \begin{equation*}
     \theta_{\text{fs}} = \beta_{1}^{0},
  \end{equation*}
  where
  \begin{equation*}
    \beta^{0} \in \arg \min_{\beta} \mathbb{E}_{\epsilon} \left[ \frac{1}{n} \sum_{i=1}^{n} (y_{i} - \sum_{k} \tilde b_{k}(t_{i},x_{i}) \beta_{k})^{2} \right].
  \end{equation*}
\end{proposition}
Thus, we have transformed the problem of estimation and inference for subpopulation or sample
incremental causal effects $\theta_{\text{fs}}$ to the problem of conducting estimation and inference of one component in a (potentially high-dimensional) linear model. Several approaches have been developed to do inference in high-dimensional linear models \citep{memepb09,liu2013asymptotic,zhang2014confidence,javanmard2014confidence,van2014asymptotically,belloni2014inference}. However, using one of these methods with the transformed data $(\tilde{\mathbf{X}},\mathbf{Y})$ does not guarantee that we still have asymptotically valid inference if the Riesz representer $-\partial_{t} \log p(t,x) \in \mathcal{B}$, but $f^{0} \not \in \mathcal{B}$ for $f^{0}(x,t) =  \mathbb{E}[Y|X=x,T=t] $. In other words, this approach does not automatically guarantee a double robustness property if the outcome model is misspecified. In the following we discuss how to obtain double robustness properties for estimating the super population effect $\theta_{sp}$. Compared to existing approaches \citep{powell1989semiparametric,hirshberg2017augmented,chernozhukov2018double}, the main difference is that the orthogonalization is done in a simple pre-processing step which allows us to rely on commonly used Lasso software  for high-dimensional regression models.
  
Define the two Lasso estimators
\begin{align*}
  \hat \gamma &:= \arg \min_{\gamma} \| \mathbf{X}_{1} - \tilde{\mathbf{X}}_{-1} \gamma \|_{2}^{2}/n + 2 \lambda_{X} \| \gamma \|_{1}, \\
  \hat \beta &:= \arg \min_{\beta} \| \mathbf{Y} - \tilde{\mathbf{X}} \beta \|_{2}^{2}/n + 2 \lambda \| \beta \|_{1}.
\end{align*}
Furthermore, define the desparsified estimator for the first component
  \begin{align*}
    \hat \beta_{1}^{\text{despar}} &= \frac{\tilde{\mathbf{Z}}^{\intercal} \mathbf{Y}}{\tilde{\mathbf{Z}}^{\intercal} \mathbf{X}_{1}} -  \sum_{k > 1}\frac{\tilde{\mathbf{Z}}^{\intercal} \tilde{\mathbf{X}}_{k}}{\tilde{\mathbf{Z}}^{\intercal} \mathbf{X}_{1}} \hat \beta_{k}, \text{ where }\\
       \tilde{\mathbf{Z}} &= \mathbf{X}_{1} - \tilde{\mathbf{X}}_{-1} \hat \gamma.
  \end{align*}
Note that this desparsified estimator is defined analogously as in the regression literature \citep{zhang2014confidence,javanmard2014confidence,van2014asymptotically,belloni2014inference}. To formulate the assumptions and the result, we need some additional notation.  
\newline
\noindent
\textbf{Notation.} Define
\begin{align*}
  \tilde{\mathbf{X}}_{k;i}^{0}&= \begin{cases} t_{i} & \text{ for $k=1$}, \\  b_{k}(t_{i},x_{i}) - t_{i} \mathbb{E}[ \partial_{t} b_{k}(t_{1},x_{1})], & \text{ for } k > 1, \end{cases} \\
  \gamma^{0} &= \arg \min_{\gamma} \mathbb{E}[\|  \mathbf{X}_{1} - \tilde{\mathbf{X}}_{-1}^{0} \gamma \|_{2}^{2}], \\
  \beta^{0} &= \arg \min_{\beta} \mathbb{E}[\| \mathbf{Y} - \tilde{\mathbf{X}}^{0} \beta \|_{2}^{2}], \\
  \tilde{\mathbf{Z}}^{0} &= \mathbf{X}_{1} -\tilde{\mathbf{X}}_{-1}^{0} \gamma^{0}, \\
  \hat \epsilon &= \mathbf{Y} - \tilde{\mathbf{X}} \hat \beta, \\
  \epsilon &= \mathbf{Y} - \tilde{\mathbf{X}}^{0} \beta^{0}.     \\
\end{align*}
\noindent
\textbf{Assumptions.}
As mentioned before, our work builds on \cite{buhlmann2015high}. In comparison to their work, the main difference in terms of assumptions is that we added assumption (A8), see below. As in our setting $\log(p)/n \rightarrow 0$, assumption (A8) means that the $\ell_{1}$-norm of the population regression coefficient $\beta^{0}$ is bounded and that the $\ell_{1}$ norm of $\gamma^{0}$ grows slower than $\sqrt{n/\log(p)}$. Thus, we consider the added assumption (A8) as rather weak. Viewed in total, the assumptions are strong and in particular require that the nonlinear functions $b_{k}$ and $\partial_{t} b_{k}$ and the error terms are bounded. 
\begin{enumerate}
\item[(A1)]  $ \mathbb{E}[(\tilde{\mathbf{X}}^{0})^{\intercal} \tilde{\mathbf{X}}^{0}]/n$ %
  has smallest eigenvalue lower bounded by $C_{1} >0$.
\item[(A2)] We assume that $b_{k}$ and $\partial_{t} b_{k}$ are bounded, i.e.\ there exists a constant $C_{2}$ such that $\mathbb{P}[|b_{k}(t_{i},x_{i})| \le C_{2}] = 1$ and the $\mathbb{P}[|\partial_{t}b_{k}(t_{i},x_{i})| \le C_{2}] = 1$ for all $k$.
\item[(A3)] $ \| \tilde{\mathbf{Z}}^{0} \|_{\infty} \le C_{3} < \infty$.
\item[(A4)] $s_{1} = |\{k : \gamma_{k}^{0} \neq 0\}| = o(\sqrt{n}/\log(p))$
\item[(A5)] $s_{0} = |\{k : \beta_{k}^{0} \neq 0\}| = o(\sqrt{n}/\log(p))$
\item[(A6)] The normalized asymptotic error is bounded from below: for \begin{align*}
  u^{2} = \text{Var} \left( \frac{\epsilon_{1} \tilde{\mathbf{Z}}_{1}^{0}}{\mathbb{E}[\tilde{\mathbf{Z}}_{1}^{0}  \tilde{\mathbf{X}}_{11}^{0}]
} + \sum_{k} \partial_{t} b_{k}(t_{1},x_{1}) \beta_{k}^{0} \right),
                                                                      \end{align*}
                                                                      we have $u^{2} \ge C_{4} > 0$.
\item[(A7)] The error is bounded $\|\epsilon\|_{\infty} \le V$.
\item[(A8)] %
  $\| \beta^{0} \|_{1}$ is bounded by a constant and $\|\gamma^{0}\|_{1} = o(\sqrt{n/\log(p)})$.
\end{enumerate}
The assumptions (A5) and (A7) can be relaxed. For details, see \cite{buhlmann2015high}. 
  Now let us turn to the result.
  \begin{theorem}%
    \label{thm:doubly-robust-conf}
 Let $(y_{i},x_{i},t_{i})$, $i=1,\ldots,n$ be i.i.d. and assume that (A1)--(A8) holds. Then, for tuning parameters $\lambda_{X} = D_{2} \sqrt{\log(p)/n}$ and $\lambda = D_{1} \sqrt{\log(p)/n}$ with constants $D_{1}, D_{2}$ sufficiently large and $\sqrt{\log(p)/n} \rightarrow 0$,
  \begin{equation}\label{eq:16}
 \frac{\sqrt{n} (\hat \beta_{1}^{\text{despar}} - \beta_{1}^{0})}{\hat u} \rightharpoonup \mathcal{N}(0,1),
\end{equation}
 where $\hat u^{2}$ is the empirical variance of
   \begin{align*}
   \frac{\hat \epsilon_{i} \tilde{\mathbf{Z}}_{i}}{(\tilde{\mathbf{Z}})^{\intercal}  \tilde{\mathbf{X}}_{1}/n
} + \sum_{k} %
      \partial_{t} b_{k}(t_{i},x_{i}) \hat \beta_{k} , \text{   for $i=1,\ldots,n$.}
   \end{align*}

\end{theorem}
Roughly speaking, this theorem allows us to construct asymptotically valid confidence intervals for the average derivative in high-dimensional settings. Note that we neither assumed $f^{0} \in \mathcal{B}$ nor $\partial_{t} \log p(t,x) \in \mathcal{B}$. The variance has two components. Loosely speaking, one part of the variance is induced by the randomness in the treatment assignment $T$, the other component comes from the randomness in $\epsilon$. The following result shows under which conditions we have $\beta_{1}^{0}=\mathbb{E}[\partial_{t} \mathbb{E}[Y|X,T]]$. The result is a variation of well-known results for doubly robust estimation of causal parameters, see for example \citet{bang2005doubly}. For reasons of completeness, we include a proof in the Appendix.
\begin{lemma}\label{lem:trick}
  If the regression surface $f^{0} = \mathbb{E}[Y|X=x,T=t] \in \mathcal{B}$ or the Riesz representer $\partial_{t} \log p(t|x) \in \mathcal{B}$, then $\beta_{1}^{0} = \mathbb{E}[\partial_{t} \mathbb{E}[Y|X,T]]$.
\end{lemma}
This shows that estimation and inference is doubly robust: if one of the two functions $f^{0}$ or $\partial_{t} \log p$ lie in $\mathcal{B}$, we obtain consistency for the average partial effect, i.e.\ $\hat \beta_{1}^{\text{despar}} \rightarrow \mathbb{E}[\partial_{t} \mathbb{E}[Y|X,T]]$ with asymptotically valid confidence intervals. If both $f^{0} \in \mathcal{B}$ and $\partial_{t} \log p(t|x) \in \mathcal{B}$ then the proposed estimator reaches the semiparametric efficiency bound \citep[Proposition 4]{newey1994asymptotic},
\begin{equation*}
  \text{Var}(\partial_{t} f^{0}) + \text{Var}(\epsilon \cdot \partial_t \log p(T|X) ).
\end{equation*}  
A proof of this result can be found in the Appendix, Lemma~\ref{lem:semieffic}.

\section{Simulations}\label{sec:sim}

In this section, we validate our theoretical results on simulated data, including settings where the assumptions of our theoretical results are violated and set-ups based on real-world data.
In Section~\ref{sec:synthetic-data-set} we discuss a very simple low-dimensional model, where the features are drawn from a Gaussian or $t$-distribution. %
In the first setting, the assumptions of Theorem~\ref{thm:variance_comp} are satisfied, whereas in the second setting the assumptions of Theorem~\ref{thm:variance_comp} are violated. In Section~\ref{sec:enhancer-data-set}, to increase the realism of our simulation study, the covariates are taken from the enhancer data set.  In Section~\ref{sec:conf} we investigate robustness of incremental treatment effects under the aforementioned low-dimensional model under varying additive confounding. Finally, in Section~\ref{sec:opp} we discuss a setting in which it is challenging to estimate incremental treatment effects. %
The results indicate that the conclusions of Theorem~\ref{thm:variance_comp} are relatively robust under violations of the assumptions and that estimating incremental effects can be challenging if the error variance is large at the edge of the observation space.

\subsection{Synthetic data set}\label{sec:synthetic-data-set}

In this section we compare estimation and statistical inference %
for the sample incremental causal effects $\theta_\text{fs}$ and the (relative) sample average treatment effect $\tau_\text{fs}$ as defined in Section~\ref{sec:smaller} in two simple scenarios. The first setting was chosen such that the assumptions of Theorem~\ref{thm:variance_comp} are satisfied, whereas the second setting was chosen such that the assumptions of Theorem~\ref{thm:variance_comp} are violated. In both cases, the regression surfaces were chosen such that they exhibit moderate curvature across the observation space. In the first scenario we generate $n$ i.i.d. observations according to the following equations:
\begin{align}
\begin{split} \label{sim1}
  h,\epsilon_{x},\epsilon_{t} & \sim \mathcal{N}(0,1)\\
  \epsilon_{y} &\sim \text{Unif}(-.5,.5) \\
    t &= h + \epsilon_{t} \\
    x &= h + \epsilon_{x} \\
    y(t) & = 3t + t^2 + x + \epsilon_{y}
\end{split}
\end{align}
We fit a cubic model $y \sim t + t^2 + t^{3} + x + x^{2}$ using ordinary least squares and then we use the simple plug-in estimators $\hat \theta$ to estimate $\theta_\text{fs}$ and $\hat \tau$ to estimate $\tau_\text{fs}$ with $t=\frac{1}{2}$, $t'=-\frac{1}{2}$ as defined in Section~\ref{sec:smaller}.

As we are fitting a linear model, asymptotically valid confidence intervals for both the incremental treatment effect and the average treatment effect can be readily computed using standard formulae. Coverage, average length of the resulting confidence intervals and root mean-squared error for varying sample size $n$ are given in Table~\ref{tab:1} %
below. Note that the model is well-specified, thus we expect asymptotically correct coverage for both estimators for their specific target quantities.
\begin{table*}%
\caption{Coverage of confidence intervals, length of confidence intervals and root mean-squared error for estimating sample incremental treatment effects and sample average treatment effects using a cubic model. The data is generated according to equation~\eqref{sim1}.\label{tab:1}
}
\centering
\begin{tabular}{rlll}
  \hline
 & n=10 & n=20 & n=50 \\ 
  \hline
CI coverage incr & 0.96 $\pm$ 0.01 & 0.96 $\pm$ 0.01 & 0.94 $\pm$ 0.01 \\ 
  CI coverage sample ATE & 0.95 $\pm$ 0.01 & 0.94 $\pm$ 0.01 & 0.93 $\pm$ 0.02 \\ 
  CI length incr & 0.87 $\pm$ 0.03 & 0.31 $\pm$ 0.01 & 0.15 $\pm$ 0 \\ 
  CI length sample  ATE & 1.09 $\pm$ 0.04 & 0.45 $\pm$ 0.01 & 0.23 $\pm$ 0 \\ 
  RMSE incr & 0.2 $\pm$ 0.01 & 0.08 $\pm$ 0 & 0.04 $\pm$ 0 \\ 
  RMSE sample  ATE & 0.25 $\pm$ 0.02 & 0.11 $\pm$ 0.01 & 0.06 $\pm$ 0 \\ 
   \hline
\end{tabular}
\end{table*}
In both cases, coverage is approximately correct. However, note that the average length of confidence intervals for the effect of the incremental effect is smaller than the respective length for the average treatment effect. This is expected and in line with our theory presented in as the assumptions for Theorem~\ref{thm:variance_comp}.
Now let us turn to a case where the assumptions of Theorem~\ref{thm:variance_comp} are violated.
We generate $n$ samples according to the following equations: %
\begin{align}
\begin{split} \label{sim2}
  h,\epsilon_{x},\epsilon_{t} & \sim t_{4}\\
  \epsilon_{y} &\sim \text{Unif}(-.5,.5) \\
    t &= h + \epsilon_{t} \\
    x &= h + \epsilon_{x} \\
    y(t) & = 3t + t^2 + x + \epsilon_{y},
\end{split}
\end{align}
where $t_4$ denotes a $t$-distribution with $4$ degrees of freedom.
Then, we fit a cubic model $y \sim t + t^2 + t^{3} + x + x^{2}$ using ordinary least squares and then use simple plug-in estimators for the sample incremental effect and the average treatment effect as above. %
Coverage, average length of the %
confidence intervals and root mean-square error are given in Table~\ref{tab:sec}. %

\begin{table*}
\caption{Coverage of confidence intervals, length of confidence intervals and root mean-squared error for estimating sample incremental treatment effects and sample average treatment effects using a cubic model. The data is generated according to equation~\eqref{sim2}.\label{tab:sec}}\centering
\begin{tabular}{rlll}
  \hline
 & n=10 & n=20 & n=50 \\ 
  \hline
CI coverage incr & 0.96 $\pm$ 0.01 & 0.96 $\pm$ 0.01 & 0.96 $\pm$ 0.01 \\ 
  CI coverage sample ATE & 0.95 $\pm$ 0.01 & 0.97 $\pm$ 0.01 & 0.95 $\pm$ 0.01 \\ 
  CI length incr & 0.66 $\pm$ 0.03 & 0.24 $\pm$ 0.01 & 0.11 $\pm$ 0 \\ 
  CI length sample  ATE & 0.87 $\pm$ 0.03 & 0.34 $\pm$ 0.01 & 0.16 $\pm$ 0 \\ 
  RMSE incr & 0.14 $\pm$ 0.01 & 0.06 $\pm$ 0 & 0.03 $\pm$ 0 \\ 
  RMSE sample  ATE & 0.18 $\pm$ 0.01 & 0.08 $\pm$ 0 & 0.04 $\pm$ 0 \\ 
   \hline
\end{tabular}
\end{table*}

Notably, %
incremental causal effects have lower root mean-squared error and lower asymptotic variance. Thus, while the assumptions for Theorem~\ref{thm:variance_comp} do not hold in this case, estimating incremental effects is still easier in terms of asymptotic mean squared-error for our choices of estimators. Of course, there exist also scenarios in which it is harder to estimate incremental than average treatment effects. This is further discussed in Section~\ref{sec:opp}.

\subsection{Simulations based on an enhancer data set}\label{sec:enhancer-data-set}

We aim to investigate estimation of the average derivative using the transformation proposed in Section~\ref{sec:high-dimens-conf}. In addition, the goal is to investigate the asymptotic efficiency result from Section~\ref{sec:smaller}.  %

To increase the realism of our simulation study,  we consider features from a real-world data set.
We consider  the activity of $36$ transcription factors in Drosophila embryos on %
$n=7809$ segments of the genome \citep{li2008transcription,macarthur2009developmental}.  As the features are heavy-tailed, a square-root transform was performed. The activity of the transcription factors is obtained using the following approach: A transcription-specific antibody is used to filter segments of DNA from the embryo. The filtered segments are measured using microarrays and mapped back to the genome, resulting in a genome-wide map of DNA binding for each transcription factor. Then, $n=7809$ segments of the genome are selected based on background knowledge about enhancer activity. The main effects and interactions of transcription factors form a $7809 \times 666$-dimensional feature matrix $\mathbf{X}$.

As a feature vector $\beta^{0} \in \mathbb{R}^{703}$, we randomly select a subset of size $26= \left \lfloor{\sqrt{703}} \right \rfloor$ of the main effects and interactions of the transcription factors, $S \subset \{1,\ldots,703\}$. For each of the effects $k \in S \cup \{703\} $, we draw $\beta_{k}^{0} \sim \exp(\lambda)$ with $\lambda = \sqrt{10}$ and $\beta_{k}^{0} = 0$ otherwise. The vector $\gamma^{0} \in \mathbb{R}^{666}$ is constructed analogously.

We draw $n$ samples $\epsilon_{i}$ and $n$ samples $\epsilon'_{i}$ from a standard Gaussian distribution with unit variance. Then, we form observations
\begin{align*}
  t_{i} &= \mathbf{X}_{i,\bullet} \gamma^{0} + \epsilon_{i}', \\
  y_{i}(t_{i}) &= \mathbf{X}_{i,\bullet} \beta_{1:666}^{0} + t_{i} \mathbf{X}_{i,\bullet} \beta_{667:702}^{0} + t_{i}^{2}  \beta_{703}^{0}  + \epsilon_{i}.
\end{align*}
We report both the mean-squared error for the average derivative $\theta_{\text{fs}}$ %
  using the method described in Section~\ref{sec:high-dimens-conf}. The tuning parameters $\lambda$ and $\lambda_{X}$ are chosen via cross-validation.
As treatment effect, we consider
\begin{equation*}
  \tau(t,t'),
\end{equation*}
where $t$ and $t'$ are both randomly drawn from the empirical distribution $(t_{i})_{i=1,\ldots,n}$.  For fairness of comparison, we also use the desparsified lasso to estimate the sample average treatment effect. %
The results can be found in Table~\ref{tab:third}. %
\begin{table*}
\caption{%
 Root mean-squared error of estimating the sample average treatment effect and sample incremental treatment effect.  The estimator of incremental causal effects exhibits consistently lower error variance. The noise is drawn from a centered Gaussian distribution with unit variance. For the feature vector, a subset of size $26= \left \lfloor{\sqrt{703}} \right \rfloor$ is randomly selected from $ S \subset \{1,\ldots,703\}$. For each of the $k \in S $, we draw $\beta_k^0 \sim \exp(\lambda)$ with $\lambda = \sqrt{10}$ and set $\beta_k^0 = 0$ otherwise.
\label{tab:third}
}
\centering
\begin{tabular}{rllll}
  \hline
 & n=200 & n=400 & n=600 & n=1000 \\ 
  \hline
RMSE incr & 0.72 $\pm$ 0.06 & 0.42 $\pm$ 0.04 & 0.22 $\pm$ 0.02 & 0.11 $\pm$ 0.01 \\ 
  RMSE subpop ATE & 0.87 $\pm$ 0.08 & 0.59 $\pm$ 0.09 & 0.29 $\pm$ 0.05 & 0.15 $\pm$ 0.04  \\
   \hline
\end{tabular}
 
\end{table*}

Evidently, in this simulation setting, estimating the average treatment effect results in higher asymptotic variance as estimating the incremental effect. This is in line with the theory presented in Section~\ref{sec:smaller}.
The Appendix contains an analogous simulation setting, where the errors are drawn from a $t$-distribution with three degrees of freedom.

\subsection{Robustness to local confounding}\label{sec:conf}

In Section~\ref{sec:sens-addit-conf}, we discussed estimation of incremental treatment effect under confounding. In this section, we investigate the behavior of the plug-in estimators $\hat \tau$ and $\hat \theta$ under confounding. Roughly speaking, we add hidden confounding $h$ so that some (but not all) subjects are confounded.  %
We call $r = \mathbb{P}[a \le T \le b]$ the ratio of confounding as it corresponds to the ratio of subjects for which the confounding is nonzero. %
We chose the simulation such that the regression surface exhibits moderate curvature. Specifically, we generate $n=100$ i.i.d.\ observations according to the following equations: %
\begin{align*}
  h,\epsilon_{t},\epsilon_{x} &\sim t_{4} \\
  \epsilon_{y} &\sim \text{Unif}(-.5,.5) \\
  t &= h + \epsilon_{t} \\
  x &= h + \epsilon_{x} \\
    y(t) &= 3t + t^2 + f_{\text{conf}}(h) + x + \epsilon_{y},
\end{align*}
where the variable $h$ only acts locally,
\begin{equation*}
  f_{\text{conf}}(h) = \begin{cases}
    0 & \text{for $h<a$}, \\
    h-a & \text{for $a \le h < b$}, \\
    b-a & \text{for $b \le h$}.
  \end{cases}
\end{equation*}
Note that the variable $h$ both affects $t$ and the outcome $y(t)$ and thus can be seen as a confounder. A cubic model $y \sim t + t^2 + t^{3} + x + x^{2}$ is fitted using ordinary least squares. We use the plug-in estimators  $\hat \theta$ and $\hat \tau$ as defined in Section~\ref{sec:smaller} to estimate $\theta_\text{fs}$ and $\tau_\text{fs}$ with $t=\frac{1}{2}$, $t'=-\frac{1}{2}$. We investigate two cases. In Figure~\ref{fig:av}, we choose the endpoints of the interval $[a,b]$ randomly and average the mean-squared error over all intervals with fixed ratio of confounded subjects $r = \mathbb{P}[a < X < b]$. In Figure~\ref{fig:max}, we consider the worst-case mean-squared error, where the maximum is taken over all intervals $[a,b]$ with the same ratio of confounded subjects $r$. In both cases, the RMSE increases under confounding and the estimator for the incremental causal effect is more robust than the estimator for the average treatment effect. Under worst-case additive confounding, the gap in RMSE between $\hat \theta$ and $\hat \tau$ widens. A bound that shows shows how estimation of incremental causal effects behaves under confounding, can be found in Section~\ref{sec:sens-addit-conf}. %

\begin{figure*}
\begin{center}
\includegraphics[scale=.7]{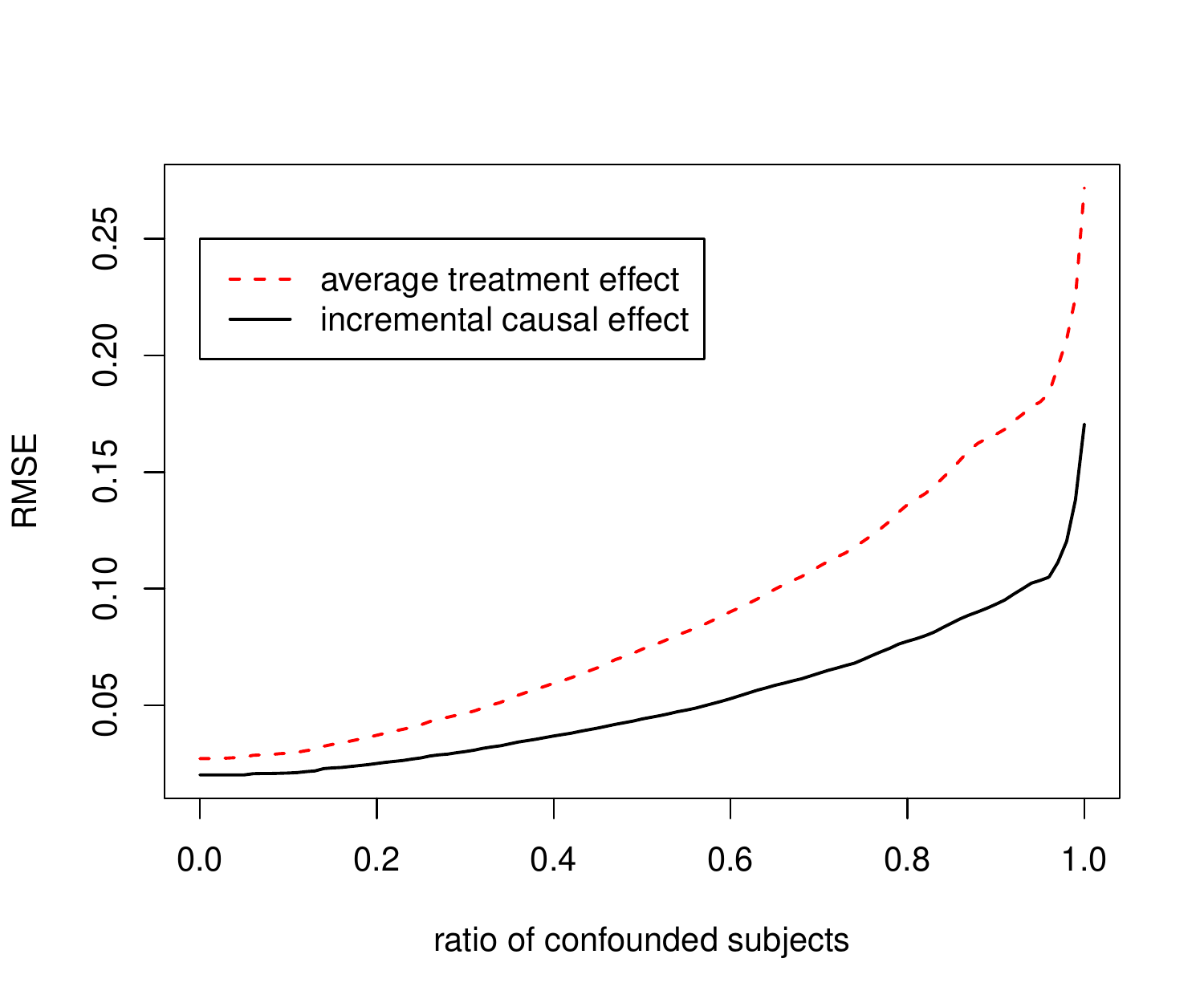}
\end{center}
\caption{Root mean-squared error under confounding. For a ratio $r \in [0,1]$ of the subjects, the ignorability assumption is violated. The plug-in estimator $\hat \theta$ for the incremental treatment  effect $\theta_{\text{fs}}$ is more robust under confounding than the plug-in estimator $\hat \tau$ for the average treatment effect $\tau_{\text{fs}}$ in this setting.} \label{fig:av}
\end{figure*}

\begin{figure*}
\begin{center}
\includegraphics[scale=.7]{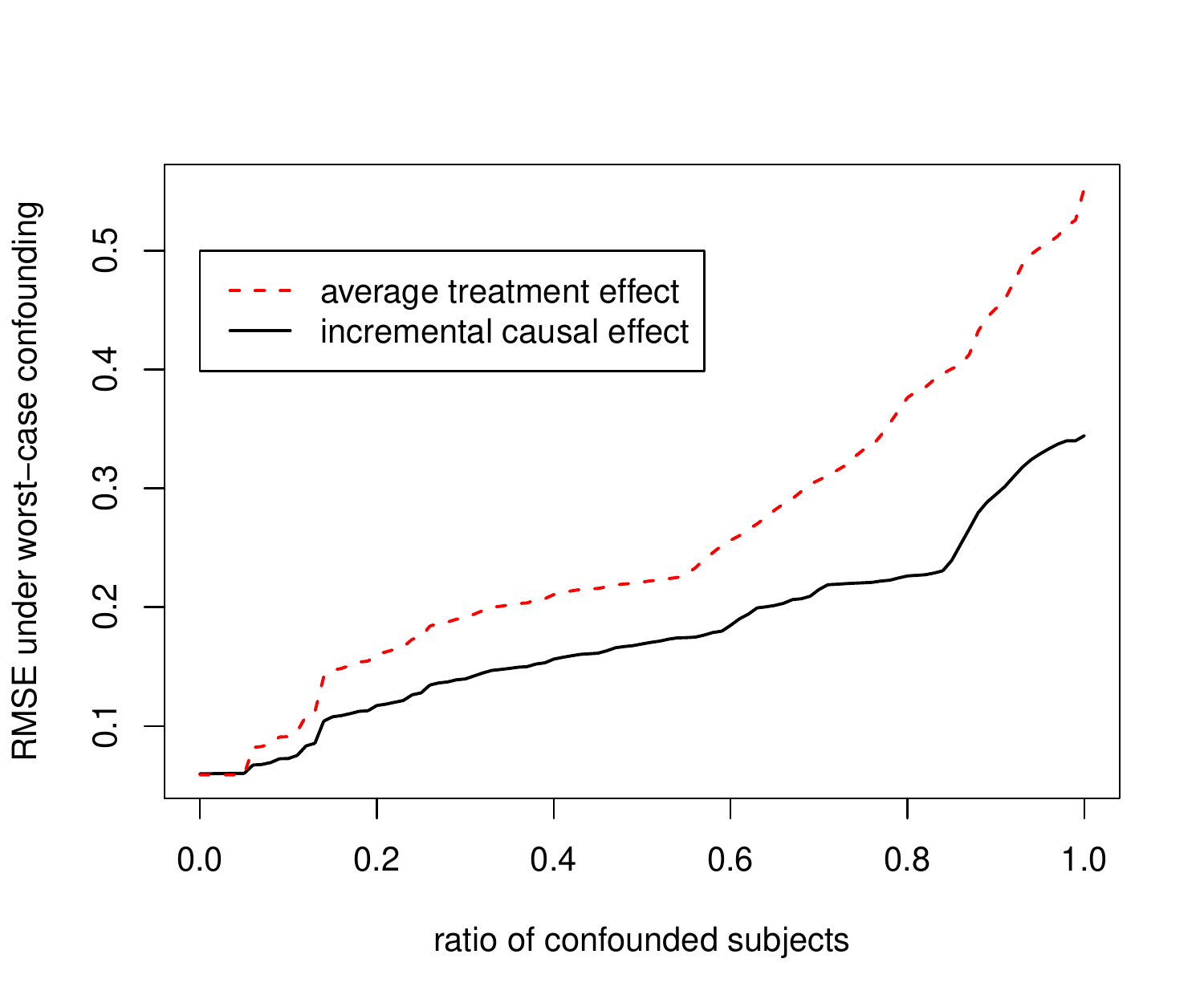}
\end{center}
\caption{Root mean-squared error under worst-case confounding. For a ratio $r \in [0,1]$ of the subjects, the contribution of the confounder $f_{\text{conf}}(h)$ is nonzero, i.e.\ the confounder only affects part of the population.  The plug-in estimator $\hat \theta$ for the incremental treatment effect $\theta_{\text{fs}}$  is more robust under confounding than the plug-in estimator $\hat \tau$ for the average treatment effect $\tau_{\text{fs}}$ in this setting. A bound that show that incremental causal effects are relatively little affected  if small parts of the population are confounded, can be found in Section~\ref{sec:sens-addit-conf}.%
}\label{fig:max}
\end{figure*}

\subsection{Challenges of estimating incremental effects}\label{sec:opp}

Of course, there exist scenarios where estimating the effect of incremental interventions is considerably harder than estimating average treatment effects. For example, performance can suffer if the error variance is larger at the tails of $T$ than in the bulk of the observations. In this case, estimating an average treatment effect $\tau(t,t')$, where $t$ and $t'$ are in the bulk of the observations, is relatively easy compared to estimating incremental causal effects. As an example, we generate i.i.d. observations according to the following equations:
\begin{align*}
    t &\sim \text{Unif}(-.5,1.5) \\
    \epsilon &\sim \text{Unif}(-.5,.5) \\
    y(t) &= t^2 + |t| \cdot\epsilon
\end{align*}
As before, a cubic model $y \sim t + t^2 + t^{3}$ is fitted using ordinary least squares. Then, the plug-in estimators estimators  $\hat \theta_{\text{fs}}$ and $\hat \tau(.5,-.5)$ are calculated as in Section~\ref{sec:smaller}. The root mean-squared error under varying sample size is reported in Table~\ref{tab:chall}. As expected, for large $n$ the mean-squared error $\mathbb{E}[(\hat \tau - \tau_{\text{fs}})^{2}]$ is smaller than the mean-squared error $\mathbb{E}[(\hat{\theta}_{\text{fs}} - \theta_{\text{fs}} )^{2}]$. If the error variance for subjects at the edge of the observation space is very large, we recommend estimating average and incremental treatment effects for a subgroup that exhibits lower error variance.
\begin{table*}
\caption{Root mean-squared error under heteroscedasticity. Estimating incremental treatment effects is difficult as the error variance at the edge of the observation space is large. In this scenario, estimating the average treatment effect $\tau(.5, -.5)$ is relatively easy as both $t=.5$ and $t'=-.5$ are in regions where the error variance is low.}\label{tab:chall}
  \centering
\begin{tabular}{rlll}
  \hline
 & n=10 & n=20 & n=50 \\ 
  \hline
RMSE sample incr & 0.43 $\pm$ 0.03 & 0.24 $\pm$ 0.01 & 0.13 $\pm$ 0.01 \\ 
  RMSE sample ATE & 1.18 $\pm$ 0.33 & 0.21 $\pm$ 0.03 & 0.09 $\pm$ 0.01 \\ 
   \hline
\end{tabular}
\end{table*}

\section{Conclusion}

The estimation of treatment effects is of central interest in many disciplines. Usually, the goal is to gather evidence to act and it might be sufficient to estimate the direction of an effect of a small (change of) action.  
Often, treatment effects are estimated under the assumption of weak ignorability and the overlap condition. %
In this paper, we have shown that these assumptions can be substantially weakened for identification of \emph{incremental} treatment effects. We introduced a local ignorability assumption and a local overlap condition and show that incremental treatment effects are identifiable under these two new local conditions. As an example, treatment assignment might be randomized locally within subgroups of patients but not across all patients. %

In simulation studies, we have seen some evidence indicating that the estimation of the average treatment effect using a plug-in estimator has often higher variance than a comparable estimator  for the incremental treatment effect. If the distribution of the treatment given the covariates is Gaussian, we have shown that this difference in asymptotic error is indeed systematic. Moreover, we derived bounds for incremental effects under confounding. %

We discussed how to obtain asymptotically valid confidence intervals that are doubly robust both in terms of estimation and inference using a two-step procedure. In the first step, a feature transformation is performed. We call this feature transformation "incremental effect orthogonalization". In the second step, an ordinary lasso regression is performed. %

Causal inference from observational data is known to be unreliable and has to be done with extreme care. In high-stake scenarios such as healthcare this can have devastating effects on human lives. We have identified situations where incremental effects can be more reliably estimated than average treatment effects. In those settings, estimating incremental effects might be more informative for practitioners than estimating the average treatment effect. %
 We hope that our work aids decisions on choice of intervention notions to reliably answer domain questions.

$\text{ }$
\newpage 

\section{Appendix}  

The Appendix contains additional simulation results and proofs.

\subsection{Additional simulation results}

Additional simulation results  can be found in Table~\ref{tab:second}.

\begin{table*}
\caption{%
The simulation setup is analogous to the one described in Section~\ref{sec:enhancer-data-set}, but the noise of the target variable follows a $t$-distribution with three degrees of freedom. Root mean-squared error of estimating the sample average treatment effect and sample incremental treatment effect. The estimator of incremental causal effects exhibits lower error.
The noise of the target variable is drawn from a $t$-distribution with three degrees of freedom. For the feature vector, a subset of size $26= \left \lfloor{\sqrt{703}} \right \rfloor$ is randomly selected from $ S \subset \{1,\ldots,703\}$. For each of the $k \in S \cup \{703 \}$, we draw $\beta_k^0 \sim \exp(\lambda)$ with $\lambda = 26$ and set $\beta_k^0 = 0$ otherwise.
\label{tab:second}
}
\centering
\begin{tabular}{rllll}
  \hline
 & n=200 & n=400 & n=600 & n=1000 \\ 
  \hline
RMSE incr & 0.72 $\pm$ 0.06 & 0.44 $\pm$ 0.04 & 0.23 $\pm$ 0.02 & 0.12 $\pm$ 0.01 \\ 
  RMSE subpop ATE & 0.91 $\pm$ 0.11 & 0.59 $\pm$ 0.11 & 0.29 $\pm$ 0.03 & 0.15 $\pm$ 0.0  \\
   \hline
\end{tabular}
 
\end{table*}

\subsection{Proof of Proposition~\ref{prop:ident}}\label{sec:proofmain}

\begin{proof}
 Without loss of generality we will drop ``conditional on $X$''. In the following, choose $t$ and $h$ with $p(t,h) > 0$.  As $Y(t)$ is continuously differentiable with derivative $Y'(t)$ there exists a random variable $\xi_\delta \in [t,t+\delta]$ such that
  \begin{equation}
      \frac{Y(t+\delta) - Y(t)}{\delta} = Y'(\xi_\delta).
  \end{equation}
  As the derivative $Y'(t)$ is continuous and bounded, by dominated convergence,
\begin{align*}
\begin{split}
  &\lim_{\delta \rightarrow 0} \frac{ \mathbb{E} \left[ Y(t+\delta)|T=t,H=h] - \mathbb{E}[Y(t) |T=t, H=h\right]}{\delta} \\
  &= \mathbb{E}[Y'(t)|T=t,H=h]. 
\end{split}
\end{align*}
 Using the ignorability assumption,
\begin{equation*}
  Y(t) \perp T | H.
\end{equation*}
Thus,
\begin{align*}
  &\mathbb{E}[Y(t+\delta)-Y(t)|T=t,H=h] \\
  &= \mathbb{E}[Y(t+\delta)|T=t,H=h]-\mathbb{E}[Y(t)|T=t,H=h]\\
  &=\mathbb{E}[Y(t+\delta)|T=t+\delta,H=h]-\mathbb{E}[Y(t)|T=t,H=h].
\end{align*}
Dividing by $\delta$,
\begin{align*}
  &\frac{\mathbb{E}[Y(t+\delta)-Y(t)|T=t,H=h]}{\delta} \\
  &= \frac{ \mathbb{E}[Y(t+\delta)|T=t+\delta,H=h]-\mathbb{E}[Y(t)|T=t,H=h]}{\delta}.
\end{align*}
Here, we used the local overlap assumption to guarantee that conditional expectations are well-defined. As shown above, for $\delta \rightarrow 0$, the limit of the quantity on the left exists and is equal to $\mathbb{E}[Y'(t)|T=t, H=h]$. Thus, $\mathbb{E}[Y|T=t,H=h]$ is differentiable and the quantity on the left converges to $\partial_{t} \mathbb{E}[Y|T=t,H=h]$. In particular, as $Y(t)$ is continously differentiable with bounded derivative, $t \mapsto \mathbb{E}[Y|T=t,H=h]$ is also continuously differentiable with bounded derivative in neighborhoods where $p(t,h) > 0$. Thus,
\begin{equation}\label{eq:31}
  \mathbb{E}[Y'(t)|H=h,T=t] = \partial_{t}  \mathbb{E}[Y(t)|H=h,T=t]
\end{equation}
Recall that under weak ignorability, for all $t$ with $p(t) >0 $ and $t'$ close to $t$ either
\begin{align*}
  & \mathbb{E}[Y|H=h',T=t'] = \mathbb{E}[Y|T=t'] \text{ or }\\
  & p(t'|h') = p(t'),
\end{align*}
for all $h'$ with $p(h'|t') >0$. In the first case,
\begin{align*}
  \partial_{t} \mathbb{E}[Y|T=t] &=  \partial_{t} \mathbb{E}[Y|H,T=t] \\
  &= \mathbb{E}[Y'(t)|H,T=t].
\end{align*}
Here, we used equation~\eqref{eq:31}. As the left-hand side of the equation is deterministic, the right-hand side is deterministic as well. Thus,
\begin{align*}
  \partial_{t} \mathbb{E}[Y|T=t] &= \mathbb{E}[Y'(t)|T=t].
\end{align*}
 In the second case, for all $t'$ close to $t$ and $h'$ with $p(h'|t') >0$,
 \begin{equation*}
   p(h'|t') = p(t'|h') \frac{p(h')}{p(t')} = p(h').
 \end{equation*}
 Thus,
 \begin{align*}
   \partial_{t} \mathbb{E}[Y|T=t] &=  \partial_{t}  \int \mathbb{E}[Y|T=t,H=h] \mathrm{d} p(h|t)  \\
                                  &= \partial_{t}  \int \mathbb{E}[Y|T=t,H=h] \mathrm{d} p(h) \\
                                  &=  \int \partial_{t} \mathbb{E}[Y|T=t,H=h] \mathrm{d} p(h) \\
                                  &=  \int \mathbb{E}[Y'(t)|T=t,H=h] \mathrm{d} p(h) \\
   &= \mathbb{E}[Y'(t)|T=t].
 \end{align*}
 In both cases,
 \begin{equation*}
   \partial_{t} \mathbb{E}[Y|T=t] = \mathbb{E}[Y'(t)|T=t].
 \end{equation*}
This concludes the proof.
\end{proof}

\subsection{Proof of Theorem~\ref{thm:variance_comp}}

\begin{proof}
  Take an orthogonal basis $b_1,\ldots,b_p$ of $\mathcal{B}$, such that $b_1 =  - \partial_t \log p(t|x) $.  For $n$ large enough, the estimator $\hat f$ can be written as $\hat f = \sum_k \hat{\alpha}_k b_k$ with unique $\hat \alpha_1,\ldots,\hat \alpha_p$.
  Define $\mathbf{X}_{j;i} = b_{j}(t_{i},x_{i})$. Note that conditionally on $\mathcal{D}_{\text{feat}}$, $\hat \theta_{\text{fs}}$ and $\hat \tau_{\text{fs}}$ are unbiased estimators of $\theta_{\text{fs}}$ and $\tau_{\text{fs}}$. Thus, in the following we will derive the conditional asymptotic variance of these estimators. The conditional variance of the vector $\hat \alpha$ can be written as
  \begin{equation*}
    (\mathbf{X}^\intercal \mathbf{X} )^{-1} \mathbb{E}[\epsilon^{2}]
  \end{equation*}
The conditional variance of $\hat \alpha$ can thus be written as %
\begin{equation*}
    \frac{1}{n} \cdot  \begin{pmatrix} 
    \frac{ \mathbb{E}[\epsilon^{2}]}{\mathbb{E}[b_1^2]} & 0  & \ldots &  0 \\
    0 & * & \ldots & * \\
    \vdots & \vdots & \ddots & \vdots \\
    0 & * & \ldots &  *
    \end{pmatrix} + o_{P} \left( \frac{1}{n}\right).
\end{equation*}
Here we used that by choice of $b_1,\ldots,b_p$ we have that $\mathbb{E}[b_1 b_k] = 0$ for all $k>1$. In addition, we used the formula for block-wise inversion and multiplication of matrices. Hence
\begin{align*}
  \hat \theta_{\text{fs}} &= \frac{1}{n} \sum_{i=1}^n \sum_k \hat \alpha_k \partial_t b_k(t_i,x_i) \\
                            &=  \sum_k \hat \alpha_k \frac{1}{n} \sum_{i=1}^n \partial_t b_k(t_i,x_i)
\end{align*}  
has asymptotic conditional variance
\begin{equation*}
  \frac{1}{n} \cdot v
  \begin{pmatrix}
    \frac{ \mathbb{E}[\epsilon^2]}{\mathbb{E}[b_1^2]}  & 0  & \ldots &  0 \\
    0 & * & \ldots & * \\
    \vdots & \vdots & \ddots & \vdots \\
    0 & * & \ldots &  *
    \end{pmatrix} v^{\intercal} +  o_{P}\left( \frac{1}{n} \right),
  \end{equation*}
  where $v =  (\mathbb{E}[\partial_t b_1],  \ldots  ,\mathbb{E}[\partial_t b_p)$.
Using that by partial integration, $\mathbb{E}[\partial_t b_k] = \mathbb{E}[b_k b_1] = 0$ for all $k > 1$, $\hat \theta_{\text{fs}}$ has asymptotic conditional variance
\begin{equation*}
    \frac{1}{n} \mathbb{E}[\partial_t b_1]^2 \frac{ \mathbb{E}[\epsilon^{2}]}{\mathbb{E}[b_1^2]} +  o_{P}\left( \frac{1}{n} \right).
\end{equation*}
We can now use that $\mathbb{E}[\partial_t b_1]^2 = \mathbb{E}[b_1^2]^2$. This gives us asymptotic conditional variance
\begin{equation*}
    \frac{\mathbb{E}[ b_1^2 ]\mathbb{E}[\epsilon^{2}]}{n} + o_{P}\left( \frac{1}{n} \right).
  \end{equation*}
Through analogous argumentation we obtain that the conditional variance of $\hat \tau$ is asymptotically
\begin{equation*}
    \frac{1}{n} \cdot \begin{pmatrix} w_1 & \ldots & w_p \end{pmatrix}  \begin{pmatrix} 
    \frac{ \mathbb{E}[\epsilon^{2}]}{\mathbb{E}[b_1^2]}  & 0  & \ldots &  0 \\
    0 & * & \ldots & * \\
    \vdots & \vdots & \ddots & \vdots \\
    0 & * & \ldots &  *
    \end{pmatrix}\begin{pmatrix} w_1 \\ \vdots \\ w_p\end{pmatrix} +  o_{P}\left( \frac{1}{n} \right),
  \end{equation*}
where $w_k := \mathbb{E}[b_k(t,X)] - \mathbb{E}[b_k(t',X)]$. Using that the submatrix denoted by ``*'' is positive semidefinite, the variance of $\hat \tau$ is asymptotically lower bounded by
\begin{equation*}
    \frac{1}{n}   \left(\mathbb{E}[b_1(t,X)] - \mathbb{E}[b_1(t',X)] \right)^2 \frac{\mathbb{E}[ \epsilon^2]}{\mathbb{E}[b_1^2]} + o_{P}\left( \frac{1}{n} \right).
  \end{equation*}
Recall that we assume that the second derivative of the log-density given $X$ is constant. In this case,
\begin{equation*}
 \left(\mathbb{E}[b_1(t,X)] - \mathbb{E}[b_1(t',X)] \right)^2 = \mathbb{E}[\partial_{1} b_{1}]^2 (t-t')^{2}.
\end{equation*}
Hence, in this case, the variance of $\hat \tau$ is asymptotically lower bounded by
\begin{equation*}
   (t-t')^{2}  \frac{\mathbb{E}[ b_1^2 ] \mathbb{E}[\epsilon^{2}]}{n} + o_{P}\left( \frac{1}{n} \right).
   \end{equation*}
Thus, $\liminf_{n \rightarrow \infty} \text{Var}(\hat \tau|\mathcal{D}_{\text{feat}}) / \text{Var}(\hat \theta_{\text{fs}} | \mathcal{D}_{\text{feat}}) \ge (t-t')^{2}$.
 \end{proof}

\subsection{Proof of Proposition~\ref{prop:robustn-addit-conf}}
\begin{proof}
 First, note that for all $(t,x)$ with $p(t,x) >0$, using Assumption~\ref{ass:regul-assumpt},
  \begin{align*}
    &\mathbb{E}[\partial_{t} \log p(T|X,H)|T=t,X=x] \\
    &= \int \frac{\partial_{t} p(t|x,h)}{p(t|x,h)} \mathrm{d} p(h|t,x)  \\
    &= \int \frac{\partial_{t} p(t|x,h)}{p(t|x)} \mathrm{d} p(h|x)   \\
    &=  \frac{1}{p(t,x)} \partial_{t} \int  p(t|x,h)  \mathrm{d} p(h|x) \\
    &=\frac{1}{p(t|x)} \partial_{t} p(t|x) \\
    &= \partial_{t} \log p(t|x).
  \end{align*}
  Using partial integration and Assumption~\ref{ass:regul-assumpt}, it follows that
  \begin{align*}
    \begin{split}
      \theta_{\text{sp}} &= \mathbb{E}[ - \partial_{t} \log p (T|X,H) Y ] \\
      &=  \mathbb{E}[ - \partial_{t} \log p (T|X,H) \mathbb{E}[Y|T,X,H]], \\
      \theta_{\text{estimated}} &=\mathbb{E}[ - \partial_{t} \log p (T|X) Y ] \\
      &= \mathbb{E}[ \mathbb{E}[- \partial_{t} \log p (T|X,H)|T,X] Y ] \\
      &= \mathbb{E}[ \mathbb{E}[- \partial_{t} \log p (T|X,H)|T,X] \mathbb{E}[Y|T,X] ] \\
      &= \mathbb{E}[ - \partial_{t} \log p (T|X,H) \mathbb{E}[Y|T,X] ].
    \end{split}
  \end{align*}
 Similarly, using partial integration and Assumption~\ref{ass:regul-assumpt},
\begin{align*}
  \theta_{\text{estimated}} &= \mathbb{E}[ - \partial_{t} \log p(T|X) Y ] \\
                            &=\mathbb{E}[ - \partial_{t} \log p(T|X) \mathbb{E}[Y|T,X,H] ] \\
  \theta_{\text{estimated}} &= \mathbb{E}[ - \partial_{t} \log p(T|X) \mathbb{E}[Y|T,X]].
\end{align*}
Combining these results,
\begin{align*}
  \theta_{\text{sp}} - \theta_{\text{estimated}} = &\mathbb{E}[(\mathbb{E}[Y|T,X,H] - \mathbb{E}[Y|T,X]) \\
 & \cdot ( -\partial_{t} \log p(T|X,H) + \partial_{t} \log p(T|X))].
\end{align*}
Using Jensen's inequality,
\begin{align}
\begin{split}  
   & | \theta_{\text{sp}} - \theta_{\text{estimated}} | \\
  \le &\mathbb{E} \large[ |\mathbb{E}[Y|T,X] - \mathbb{E}[Y|T,X,H]| \cdot \\
  & |\partial_{t} \log p(T|X) - \partial_{t} \log p(T|X,H)| \large].
\end{split} 
\end{align}
This concludes the proof.  
\end{proof}

\subsection{Proof of Theorem~\ref{thm:doubly-robust-conf} and auxiliary results}
The proof follows closely \citet{buhlmann2015high} with some modifications.
 Before we proceed, we show that the following auxiliary results hold:
\begin{itemize}
\item[(D1)]
  \begin{align*}
    \max_{k } |\epsilon^{\intercal} \tilde{\mathbf{X}}_{k}/n| &=  O_{P}(\sqrt{\log(p)/n}) \\
    \max_{k } | \epsilon^{\intercal} (\tilde{\mathbf{X}}_{k} - \tilde{\mathbf{X}}_{k}^{0})/n| &= O_{P}(\sqrt{\log(p)}/n) \\
    \max_{k \neq 1} | (\tilde{\mathbf{X}}_{k}^{\intercal} \tilde{\mathbf{Z}}^{0})/n| &= O_{P}(\sqrt{\log(p)/n})\\
    \max_{k \neq 1} | (\tilde{\mathbf{X}}_{k} - \tilde{\mathbf{X}}_{k}^{0})^{\intercal} \tilde{\mathbf{Z}}^{0}/n| &= O_{P}(\sqrt{\log(p)}/n)
    \end{align*}
\item[(D2)] $\| \hat \gamma(\lambda_X) - \gamma^{0} \|_{1} = o_{P}(1/\sqrt{\log(p)})$
\item[(D3)] $ \|\hat \beta(\lambda) - \beta^{0} \|_{1} = o_{P}(1/\sqrt{\log(p)})$ 
\end{itemize}

\begin{lemma}[Similar to Lemma 1 in \cite{buhlmann2015high}]\label{le:error}
Assume (A2), (A3) and (A7). Then, (D1) holds.
\end{lemma}

\begin{proof}
We will prove the first part of the statement. The other parts of the statement can be proven analogously. First, note that
\begin{align}\label{eq:8}
  \begin{split}
  &\mathbb{E}[\max_{1 \le k \le p} |n^{-1} \epsilon^{\intercal} \tilde{\mathbf{X}}_{k} |^{2}] \\
  \le \, \, &2 \mathbb{E}[\max_{1 \le k \le p} |n^{-1} \epsilon^{\intercal} \mathbf{X}_{k} |^{2}] +  2 \mathbb{E}[\max_{1 \le k \le p} | \delta_{k} |^{2}],
\end{split}  
\end{align}
where $\delta_{1} = 0$ and $\delta_{k} = n^{-1} \sum_{i} \epsilon_{i} t_{i}
\cdot n^{-1}\sum_{i}\partial_{t} b_{k}(t_{i},x_{i}) $ for $k>1$. Using Nemirowski's inequality \cite[Lemma 14.24]{pbvdg11} %
we obtain:
\begin{equation}\label{eq:6}
 \mathbb{E}[\max_{1 \le k \le p} |n^{-1} \epsilon^{\intercal} \mathbf{X}_{k} |^{2}] \le 8 \log(2p) C_{2}^{2} V^{2}/n = O(\log(p)/n),
\end{equation}
and similarly
\begin{equation}\label{eq:7}
 \mathbb{E}[\max_{1 \le k \le p} | \delta_{k} |^{2}] \le 8 \log(2p) C_{2}^{4} V^{2}/n = O(\log(p)/n).
\end{equation}
Using equation~\eqref{eq:6} and equation~\eqref{eq:7} in equation~\eqref{eq:8} we obtain
\begin{equation}
\mathbb{E}[\max_{1 \le k \le p} |n^{-1} \epsilon^{\intercal} \tilde{\mathbf{X}}_{k} |^{2}]  = O(\log(p)/n).
\end{equation}
In the next step, we can use Markov's inequality and $\mathbb{E}[\epsilon^{\intercal} \tilde{\mathbf{X}}_{k}] = 0$ to conclude that
\begin{align*}
  \mathbb{P}[\max_{k=1,\ldots,p}|n^{-1} \epsilon^{\intercal} \tilde{\mathbf{X}}_{k} | > c] &\le \frac{\mathbb{E}[ \max_{k=1,\ldots,p} |n^{-1} \epsilon^{\intercal} \tilde{\mathbf{X}}_{k} |]}{c} \\
                                                              &\le \frac{\sqrt{\mathbb{E}[ \max_{k=1,\ldots,p} |n^{-1} \epsilon^{\intercal} \tilde{\mathbf{X}}_{k} |^{2}]}}{c} \\ &= O(\sqrt{\log(p)/n}).
\end{align*}

\end{proof}

\begin{lemma}[Similar to Lemma 2 in \cite{buhlmann2015high}]\label{lem:D2D3}
  Assume (A1) and (A2) and $\sqrt{\log(p)/n} \rightarrow 0$.
  \begin{enumerate}
   \item If (A3) and (A4) hold, then for $\lambda_{X} = D_{2} \sqrt{\log(p)/n}$ with $D_{2}$ sufficiently large we have (D2).
   \item If (A7) and (A5) holds, then for $\lambda = D_{1} \sqrt{\log(p)/n}$ with $D_{1}$ sufficiently large, we have (D3).
   \end{enumerate}
\end{lemma}

\begin{proof}
  The proof proceeds mostly as in \cite{buhlmann2015high}.  However, there is one slight complication that  the rows of $\tilde{\mathbf{X}}$ are not i.i.d. We will prove statement (1). The proof for statement (2) proceeds analogously. 

First, we will prove the compatibility condition for the transformed data $\tilde{\mathbf{X}}$. To this end, note that 
\begin{align}\label{eq:9}
\begin{split}  
  \frac{1}{n}\tilde{\mathbf{X}}_{j}^{\intercal} \tilde{\mathbf{X}}_{k} = &  \frac{1}{n} \mathbf{X}_{j}^{\intercal}  \mathbf{X}_{k} -  \frac{1}{n} \sum_{i=1}^{n} \mathbf{X}_{i,j} \mathbf{X}_{i,1} \delta_{k} \\
&  - \frac{1}{n} \sum_{i=1}^{n}\mathbf{X}_{i,1} \mathbf{X}_{i,k} \delta_{j} + \delta_{k} \delta_{j} \frac{1}{n} \sum_{i=1}^{n}\mathbf{X}_{i,1}^{2},
\end{split}    
  \end{align}
  where $\delta_{j} =  \frac{1}{n} \sum_{i'} \partial_{t} b_{j}(t_{i'},x_{i'}) $ for $j>1$ and $\delta_{j}=0$ for $j=1$. Using assumption (A2) and sub-Gaussian tail bounds \citep[Chapter 2]{boucheron2013concentration}, the terms
  \begin{align*}
    &\frac{1}{n} \mathbf{X}_{j}^{\intercal} \mathbf{X}_{k}  - \frac{1}{n}\mathbb{E}[\mathbf{X}_{j}^{\intercal} \mathbf{X}_{k}], \\
    & \delta_{k} - \mathbb{E}[\delta_{k}], %
  \end{align*}
are uniformly of the order $O_{P}(\sqrt{\log(p)/n})$. Hence, using equation \eqref{eq:9}, the term
  \begin{equation*}
  \max_{j,k} \left| \frac{1}{n}\tilde{\mathbf{X}}_{j}^{\intercal} \tilde{\mathbf{X}}_{k} -  \frac{1}{n} \mathbb{E}[(\tilde{\mathbf{X}}_{j}^{0})^{\intercal} \tilde{\mathbf{X}}_{k}^{0}] \right| %
  \end{equation*}
  is of order $O_{P}(\sqrt{\log(p)/n})$.
  
  The sparsity assumption (A4) %
  combined with (A1) imply that the compatibility condition holds with probability converging to one, c.f.\ \citet[Chapter 6.12]{pbvdg11}. Using Lemma~\ref{le:error}, we obtain $\| \tilde{\mathbf{X}}_{-1}^{\intercal} \tilde{\mathbf{Z}}^{0} \|_{\infty} \le O_{P}(\sqrt{n \log(p)})$.  Invoking an inequality for the lasso \citep[Chapter 6.1]{pbvdg11} with assumption (A4), we obtain statement~(1).
\end{proof}

\begin{proposition}[Similar to Proposition 7 in \cite{buhlmann2015high}]\label{prop-defw}
  Assume (A1), (A2), (A3), (A6), (A7) and (A8). Write
\begin{align*}
  u^{2} = \text{Var} \left( \frac{\epsilon_{1} \tilde{\mathbf{Z}}_{1}^{0}}{\mathbb{E}[\tilde{\mathbf{Z}}_{1}^{0}  \tilde{\mathbf{X}}_{11}^{0}]
} + \sum_{k}  \partial_{t} b_{k}(t_{1},x_{1}) \beta_{k}^{0}) \right)
\end{align*}
Then,
\begin{align}\label{eq:12}
  \begin{split}
   & \sqrt{n} \left( \frac{ \frac{\epsilon^{\intercal} \tilde{\mathbf{Z}}^{0}/n}{\mathbb{E}[\tilde{\mathbf{Z}}^{0}_{1}  \tilde{\mathbf{X}}_{11}^{0}]
        } - \sum_{k} (\mathbb{E}[\partial_{t} b_{k} \beta_{k}^{0} ] - \hat{\mathbb{E}}[\partial_{t} b_{k} \beta_{k}^{0} ] ) }{ u }\right)  \\
    &\rightharpoonup \mathcal{N}(0,1).
\end{split}
  \end{align}

\end{proposition}

\begin{proof}
By assumption (A6), $u$ is bounded from below. In addition, note that due to (A1), $\mathbb{E}[(\tilde{\mathbf{Z}}^{0})^{\intercal}  \tilde{\mathbf{X}}_{1}^{0}]/n$ is bounded away from zero and due to (A2) and (A3) it is bounded from above. Due to (A2) and (A8), $\sum_{k} \partial_{t} b_{k} \beta_{k}^{0}$ is bounded. The proof then proceeds analogously to the proof in \cite{buhlmann2015high} using the Lindeberg condition.
\end{proof}

\begin{proposition}[Similar to Proposition 8 in \cite{buhlmann2015high}]\label{prop-defw2}
  Assume $\sqrt{\log(p)/n} \rightarrow 0$, (A1), (A2), (A3), (A6), (A7), (A8), (D1), (D2) and (D3). Then:
\begin{equation}\label{eq:13}
  \sqrt{n} \left( \frac{\frac{\tilde{\mathbf{Z}}^{\intercal} \epsilon}{\tilde{\mathbf{Z}}^{\intercal} \tilde{\mathbf{X}}_{1}} -  \sum_{k} ( \mathbb{E}[\partial_{t} b_{k} \beta_{k}^{0} ] - \hat{\mathbb{E}}[\partial_{t}  b_{k} \beta_{k}^{0}] )}{u}  \right) \rightharpoonup \mathcal{N}(0,1)
\end{equation}
\end{proposition}

\begin{proof}
  We have to show that the difference between equation~\eqref{eq:12} and equation~\eqref{eq:13} is of order $o_{P}(1)$. Note that due to (A6), the quantity $u$ is bounded away from zero and can be ignored. The difference between equation~\eqref{eq:12} and equation~\eqref{eq:13}, up to bounded factors, is
\begin{equation*}
  \sqrt{n} \left( \frac{\tilde{\mathbf{Z}}^{\intercal} \epsilon}{\tilde{\mathbf{Z}}^{\intercal} \tilde{\mathbf{X}}_{1}}   - \frac{\epsilon^{\intercal} \tilde{\mathbf{Z}}^{0}/n}{\mathbb{E}[(\tilde{\mathbf{Z}}_{1}^{0})^{\intercal}  \tilde{\mathbf{X}}_{11}^{0}] } \right).
\end{equation*}
We want to show that this terms goes to zero.
Let us assume for a moment that
\begin{enumerate}
 \item[(1)]  $| \epsilon^{\intercal}(\tilde{\mathbf{Z}}^{0} - \tilde{\mathbf{Z}}) / \sqrt{n}| = o_{P}(1)$,
 \item[(2)] $ \tilde{\mathbf{Z}}^{\intercal} \tilde{\mathbf{X}}_{1}/n - \mathbb{E}[(\tilde{\mathbf{Z}}_{1}^{0})^{\intercal} \tilde{\mathbf{X}}_{11}^{0}] = o_{P}(1) $,
 \item[(3)]  $ \epsilon^{\intercal} \tilde{\mathbf{Z}}^{0}/\sqrt{n} = O_{P}(1)$,
\item[(4)]  $\mathbb{E}[(\tilde{\mathbf{Z}}_{1}^{0})^{\intercal} \tilde{\mathbf{X}}_{11}^{0}]$ is bounded away from zero.
\end{enumerate}
Then,
\begin{align*}
 & \, \, \sqrt{n} \left( \frac{\tilde{\mathbf{Z}}^{\intercal} \epsilon}{\tilde{\mathbf{Z}}^{\intercal} \tilde{\mathbf{X}}_{1}}   - \frac{\epsilon^{\intercal} \tilde{\mathbf{Z}}^{0}/n}{\mathbb{E}[(\tilde{\mathbf{Z}}_{1}^{0})^{\intercal}  \tilde{\mathbf{X}}_{11}^{0}] } \right)\\ &= \epsilon^{\intercal} \tilde{\mathbf{Z}}^{0}/\sqrt{n} \left(\frac{1}{\tilde{\mathbf{Z}}^{\intercal} \tilde{\mathbf{X}}_{1}/n} - \frac{1}{\mathbb{E}[(\tilde{\mathbf{Z}}_{1}^{0})^{\intercal}  \tilde{\mathbf{X}}_{11}^{0}]}\right) \\
                                                                                                                                                                                                                                                                             &\, \,+ ( \epsilon^{\intercal} \tilde{\mathbf{Z}}^{0}/\sqrt{n} - \epsilon^{\intercal} \tilde{\mathbf{Z}}/\sqrt{n} )   \frac{1}{\tilde{\mathbf{Z}}^{\intercal} \tilde{\mathbf{X}}_{1}/n} \\
  &= o_{P}(1),
\end{align*}
which is the desired result. Thus, it remains to show that the claims (1)--(4) hold. Let us first show claim (1).
\begin{align*}
  | \epsilon^{\intercal}(\tilde{\mathbf{Z}}^{0} - \tilde{\mathbf{Z}}) / \sqrt{n}| &\le | \epsilon^{\intercal} \tilde{\mathbf{X}}_{-1} (\hat \gamma - \gamma^{0})|/\sqrt{n} \\
                                                                                    &+ |\epsilon^{\intercal} (\tilde{\mathbf{X}}_{-1} - \tilde{\mathbf{X}}_{-1}^{0}) \gamma^{0}|/\sqrt{n} \\
  &\le \| \epsilon^{\intercal} \tilde{\mathbf{X}}_{-1} \|_{\infty} \| \hat \gamma -  \gamma^{0} \|_{1} / \sqrt{n} \\ &+  \|\epsilon^{\intercal} (\tilde{\mathbf{X}}_{-1} - \tilde{\mathbf{X}}_{-1}^{0}) \|_{\infty} / \sqrt{n} \| \gamma^{0} \|_{1}
\end{align*}
Now we can use (D1), (D2) and (A8) to conclude that this term goes to zero in probability for $n \rightarrow \infty$.  This proves claim (1). Now let us turn to claim (2).
Similarly as proving claim (1) we can show that
\begin{align*}
  |\tilde{\mathbf{Z}}^{\intercal} \tilde{\mathbf{X}}_{1}/n - (\tilde{\mathbf{Z}}^{0})^{\intercal} \tilde{\mathbf{X}}_{1}/n | &= | \frac{1}{n} \tilde{\mathbf{X}}_{1}^{\intercal}  \tilde{\mathbf{X}}_{-1} (\hat \gamma - \gamma^{0})| \\& + | \frac{1}{n} \tilde{\mathbf{X}}_{1}^{\intercal} (\tilde{\mathbf{X}}_{-1}^{0} - \tilde{\mathbf{X}}_{-1} ) \gamma^{0} |\\
  &= o_{P}(1).
\end{align*}
As $\tilde{\mathbf{X}}_{1} = \tilde{\mathbf{X}}_{1}^{0}$,
\begin{equation}\label{eq:19}
   |\tilde{\mathbf{Z}}^{\intercal} \tilde{\mathbf{X}}_{1}/n - (\tilde{\mathbf{Z}}^{0})^{\intercal} \tilde{\mathbf{X}}_{1}^{0}/n | = o_{P}(1).
\end{equation}
By (A2) and the law of large numbers,
\begin{equation*}
  (\tilde{\mathbf{Z}}^{0})^{\intercal} \tilde{\mathbf{X}}_{1}^{0}/n - \mathbb{E}[(\tilde{\mathbf{Z}}_{1}^{0})^{\intercal}  \tilde{\mathbf{X}}_{11}^{0}] = o_{P}(1).
\end{equation*}
Using equation~\eqref{eq:19} proves claim (2).
\begin{equation*}
  \left(\frac{1}{\tilde{\mathbf{Z}}^{\intercal} \tilde{\mathbf{X}}_{1}/n} - \frac{1}{\mathbb{E}[(\tilde{\mathbf{Z}}_{1}^{0})^{\intercal}  \tilde{\mathbf{X}}_{11}^{0}]}\right) = o_{P}(1).
\end{equation*}
Due to (A3), (A7) and the definition of $\epsilon$, $ \epsilon^{\intercal} \tilde{\mathbf{Z}}^{0}/\sqrt{n} = O_{P}(1)$. This proves claim (3). Claim (4) follows from assumption (A1).  %
\end{proof}

\begin{proposition}[Similar to Proposition 9 in \cite{buhlmann2015high}]\label{prop:almost}
  Assume (A1), (A2), (A3), (A6), (A7), (A8), (D1), (D2) and (D3). Then, for $\lambda_{X} = D_{2} \sqrt{\log(p)/n}$ with $D_{2}$ sufficiently large and as $$\sqrt{\log(p)/n}\rightarrow 0,$$ we have
\begin{equation*}
  \frac{\sqrt{n}(\hat \beta_{1}^{\text{despar}} - \beta_{1}^{0})}{u} \rightharpoonup \mathcal{N}(0,1),
\end{equation*}
where $u$ is defined as in Proposition~\ref{prop-defw}. %
\end{proposition}

\begin{proof}
  Let us recall the definition
  \begin{equation*}
    \hat \beta_{1}^{\text{despar}} = \frac{\tilde{\mathbf{Z}}^{\intercal} \mathbf{Y}}{\tilde{\mathbf{Z}}^{\intercal} \tilde{\mathbf{X}}_{1}} -  \sum_{k \neq 1}\frac{\tilde{\mathbf{Z}}^{\intercal} \tilde{\mathbf{X}}_{k}}{\tilde{\mathbf{Z}}^{\intercal} \tilde{\mathbf{X}}_{1}} \hat \beta_{k}.
  \end{equation*}
  Then, using that $$\mathbf{Y} = \epsilon + \tilde{\mathbf{X}}^{0} \beta^{0} = \epsilon + \tilde{\mathbf{X}} \beta^{0} - \sum_{k \neq 1}\tilde{\mathbf{X}}_{1} (\mathbb{E}[\partial_{t} b_{k} \beta_{k}^{0}]-\hat{\mathbb{E}}[\partial_{t} b_{k} \beta_{k}^{0}]),$$ we have
  \begin{align*}
    & \, \, \sqrt{n} \frac{\tilde{\mathbf{Z}}^{\intercal} \tilde{\mathbf{X}}_{1}}{n} (\hat \beta_{1}^{\text{despar}} - \beta_{1}^{0}) \\
    &= \frac{1}{\sqrt{n}} \left( \tilde{\mathbf{Z}}^{\intercal} \mathbf{Y} - \sum_{ k \neq 1} \tilde{\mathbf{Z}}^{\intercal} \tilde{\mathbf{X}}_{k} \hat \beta_{k} - \tilde{\mathbf{Z}}^{\intercal} \tilde{\mathbf{X}}_{1} \beta_{1}^{0} \right) \\
    &= \frac{1}{\sqrt{n}} \large( \tilde{\mathbf{Z}}^{\intercal} (\epsilon - \tilde{\mathbf{X}}_{1} \sum_{k \neq 1} ( \mathbb{E}[\partial_{t} b_{k}] - \hat{\mathbb{E}}[\partial_{t} b_{k}] ) \beta_{k}^{0} \\
      &+ \sum_{k \neq 1} \tilde{\mathbf{X}}_{k} (\beta_{k}^{0} - \hat \beta_{k})) \large) \\
    &= \frac{1}{\sqrt{n}} \large( \tilde{\mathbf{Z}}^{\intercal} (\epsilon - \tilde{\mathbf{X}}_{1} \sum_{k \neq 1} ( \mathbb{E}[\partial_{t} b_{k}] - \hat{\mathbb{E}}[\partial_{t} b_{k}] ) \beta_{k}^{0} \\&+  \tilde{\mathbf{X}}_{-1} (\beta_{-1}^{0} - \hat \beta_{-1})) \large).
  \end{align*}
  The latter quantity in this term can be bounded,
  \begin{align*}
    & \, \, \left| \frac{1}{\sqrt{n}} \tilde{\mathbf{Z}}^{\intercal} \tilde{\mathbf{X}}_{-1} (\beta_{-1}^{0} - \hat \beta_{-1}) \right| \\
   &\le  \frac{1}{\sqrt{n}} \left\| \tilde{\mathbf{Z}}^{\intercal} \tilde{\mathbf{X}}_{-1} \right\|_{\infty} \left\| (\beta_{-1}^{0} - \hat \beta_{-1})  \right\|_{1}  \\
  \end{align*}
  The KKT conditions for the regression of $\tilde{\mathbf{X}}_{1}$ on $\tilde{\mathbf{X}}_{-1}$ read as
  \begin{equation*}
    \tilde{\mathbf{X}}_{-1}^{\intercal} \tilde{\mathbf{Z}}/n + \lambda_{X} \hat \kappa = 0,
  \end{equation*}
for $\hat \kappa \in [-1,1]^{p-1}$. Thus, $\| \tilde{\mathbf{X}}_{-1}^{\intercal} \tilde{\mathbf{Z}} \|_{\infty} = O(\sqrt{n\log(p)}) $. Furthermore, by assumption $ \| \beta^{0} - \hat \beta \|_{1} = o_{P}(1/\sqrt{\log(p)})$. Thus,
  \begin{equation*}
     \, \, \left| \frac{1}{\sqrt{n}} \tilde{\mathbf{Z}}^{\intercal} \tilde{\mathbf{X}}_{-1} (\beta_{-1}^{0} - \hat \beta_{-1}) \right| = o_{P}(1)
   \end{equation*}
   Thus,
   \begin{align*}
      & \, \, \sqrt{n} \frac{\tilde{\mathbf{Z}}^{\intercal} \tilde{\mathbf{X}}_{1}}{n} (\hat \beta_{1}^{\text{despar}} - \beta_{1}^{0}) \\
    &= \frac{1}{\sqrt{n}} \left( \tilde{\mathbf{Z}}^{\intercal} (\epsilon - \tilde{\mathbf{X}}_{1} \sum_{k \neq 1} ( \mathbb{E}[\partial_{t} b_{k}] - \hat{\mathbb{E}}[\partial_{t} b_{k}] ) \beta_{k}^{0}) \right) + o_{P}(1).
   \end{align*}
   Rearranging and using property (2) and (4) from the proof of Proposition~\ref{prop-defw2}  yields
   \begin{align*}
      & \, \, \sqrt{n} (\hat \beta_{1}^{\text{despar}} - \beta_{1}^{0}) \\
      &= \sqrt{n} \left( \frac{\tilde{\mathbf{Z}}^{\intercal} \epsilon}{\tilde{\mathbf{Z}}^{\intercal} \tilde{\mathbf{X}}_{1}} - \sum_{k \neq 1} ( \mathbb{E}[\partial_{t} b_{k}] - \hat{\mathbb{E}}[\partial_{t} b_{k}] ) \beta_{k}^{0} \right) + o_{P}(1).%
        \end{align*}
Using Proposition~\ref{prop-defw2} completes the proof.
\end{proof}

\begin{proposition}[Similar to Proposition 1 in \cite{buhlmann2015high}]\label{prop:varianceestimation}
 Assume $\sqrt{\log(p)/n} \rightarrow 0$, (A1), (A2), (A3), (A6), (A7), (A8), (D2) and (D3). Then,
 \begin{align*}
   \hat u^{2} = u^{2} +o_{P}(1),
 \end{align*}
 where $\hat u^{2}$ is the empirical variance of
   \begin{align*}
   \frac{\hat \epsilon_{i} \tilde{\mathbf{Z}}_{i}}{(\tilde{\mathbf{Z}})^{\intercal}  \tilde{\mathbf{X}}_{1}/n
} - \sum_{k} \hat{\mathbb{E}}[\partial_{t} b_{k}  ]\hat \beta_{k} - \partial_{t} b_{k}(t_{i},x_{i}) \hat \beta_{k} ,
   \end{align*}
   for $i=1,\ldots,n$ and $u^{2}$ is the variance of
     \begin{align*}
    \frac{\epsilon_{1} \tilde{\mathbf{Z}}_{1}^{0}}{\mathbb{E}[(\tilde{\mathbf{Z}}_{1}^{0})^{\intercal}  \tilde{\mathbf{X}}_{1;1}^{0}]
} - \sum_{k} \mathbb{E}[\partial_{t} b_{k} \beta_{k}^{0} ] - \partial_{t} b_{k}(t_{1},x_{1}) \beta_{k}^{0}.
     \end{align*}
   \begin{proof}
Define    
  \begin{align*}
  \xi_{i}^{0} =  \frac{\epsilon_{i} \tilde{\mathbf{Z}}_{i}^{0}}{\mathbb{E}[(\tilde{\mathbf{Z}}_{1}^{0})^{\intercal}  \tilde{\mathbf{X}}_{1;1}^{0}]
} - \sum_{k} \mathbb{E}[\partial_{t} b_{k} \beta_{k}^{0} ] - \partial_{t} b_{k}(t_{i},x_{i}) \beta_{k}^{0},
  \end{align*}
    and
    \begin{align*}
  \xi_{i} =  \frac{\hat \epsilon_{i} \tilde{\mathbf{Z}}_{i}}{(\tilde{\mathbf{Z}})^{\intercal}  \tilde{\mathbf{X}}_{1}/n
} - \sum_{k} \hat{\mathbb{E}}[\partial_{t} b_{k}  ]\hat \beta_{k} - \partial_{t} b_{k}(t_{i},x_{i}) \hat \beta_{k}.
\end{align*}
  By assumption (A1), $\mathbb{E}[(\tilde{\mathbf{Z}}_{1}^{0})^{\intercal}  \tilde{\mathbf{X}}_{1;1}^{0}]$ is bounded away from zero. Thus, by (A2), (A3), (A7) and (A8), the $\xi_{i}^{0}$ are bounded. Using the law of large numbers,
  \begin{align*}
    \frac{1}{n} \sum_{i} \xi_{i}^{0} &= \mathbb{E}[\xi_{1}^{0}] + o_{P}(1), \\
    \frac{1}{n} \sum_{i} (\xi_{i}^{0})^{2} &= \mathbb{E}[(\xi_{1}^{0})^{2}] + o_{P}(1).
  \end{align*}
  Thus, it suffices to show that
  \begin{align*}
    \frac{1}{n} \sum_{i} \xi_{i}^{0} - \xi_{i} &=  o_{P}(1), \\
    \frac{1}{n} \sum_{i}  (\xi_{i}^{0})^{2} - \xi_{i}^{2} &=  o_{P}(1).
  \end{align*}
  Note that we have
  \begin{align}\label{eq:21}
    \begin{split}
    \left| \frac{1}{n} \sum_{i} \xi_{i}^{0} - \xi_{i} \right|| &\le \max_{i} \left| \xi_{i}^{0} - \xi_{i} \right| \\
     \left| \frac{1}{n} \sum_{i}  (\xi_{i}^{0})^{2} - \xi_{i}^{2} \right| &\le \max_{i} |  \xi_{i}^{0} - \xi_{i} |  (\max_{i} | \xi_{i}^{0} - \xi_{i} | + \max_{i} | \xi_{i}^{0} |) 
    \end{split}
  \end{align}
  As the $\xi_{i}^{0}$ are bounded, using equation~\eqref{eq:21} it suffices to show that
  \begin{equation}
    \max_{i}| \xi_{i} - \xi_{i}^{0}| = o_{P}(1).
  \end{equation}
  We will do this in two steps.
  \begin{align}\label{eq:24}
    \begin{split}
    & \, \, \max_{i}| \xi_{i} - \xi_{i}^{0}|  \\
    &\le \max_{i} \left|  \frac{\epsilon_{i} \tilde{\mathbf{Z}}_{i}^{0}}{\mathbb{E}[(\tilde{\mathbf{Z}}_{1}^{0})^{\intercal}  \tilde{\mathbf{X}}_{1;1}^{0}]
      } -   \frac{\hat \epsilon_{i} \tilde{\mathbf{Z}}_{i}}{(\tilde{\mathbf{Z}})^{\intercal}  \tilde{\mathbf{X}}_{1}/n }\right|
       \\
   &   + \max_{i} \large| \sum_{k} \mathbb{E}[\partial_{t} b_{k} \beta_{k}^{0} ] - \partial_{t} b_{k}(t_{i},x_{i}) \beta_{k}^{0} 
     -  \sum_{k} \hat{\mathbb{E}}[\partial_{t} b_{k}  ]\hat \beta_{k} \\
     &- \partial_{t} b_{k}(t_{i},x_{i}) \hat \beta_{k} \large|.
   \end{split}
  \end{align}
  By assumption, $\xi_{i}^{0}$ and $\xi_{i}$ are bounded. 
Note that
    \begin{align*}
      &\max_{i}  \left| \sum_{k} \partial_{t} b_{k}(t_{i},x_{i}) (\beta_{k}^{0}- \hat \beta_{k} ) \right| \\
      &\le \max_{i} \max_{k} |b_{k}(t_{i},x_{i})|^{2} \|\beta^{0} - \hat \beta \|_{1}.
    \end{align*}
    Due to assumption (A2), the $b_{k}$ are bounded. %
    Recall that due to (D3), $\| \hat \beta - \beta^{0} \|_{1} = o_{P}(\sqrt{1/\log{p}})$. Thus,
    \begin{align}\label{eq:22}
    \max_{i}  \left| \sum_{k} \partial_{t} b_{k}(t_{i},x_{i}) (\beta_{k}^{0}- \hat \beta_{k} ) \right| = o_{P}(1).
    \end{align}
   Using (A2) and that $\| \beta^{0} \|_{1}$ is bounded, using a sub-Gaussian tail inequality \citep[Chapter 2]{boucheron2013concentration},
   \begin{align}\label{eq:23}
     \begin{split}
      &\left| \sum_{k} \mathbb{E}[\partial_{t} b_{k}] \beta_{k}^{0}  - \sum_{k} \frac{1}{n} \sum_{i} \partial_{t} b_{k}(t_{i},x_{i}) \beta_{k}^{0} \right| \\
      = \, \, & O_{P}(\sqrt{\log(p)/n}).
      \end{split}
    \end{align}
    Combining equation~\eqref{eq:22} and equation~\eqref{eq:23}, as $\sqrt{\log(p)/n} \rightarrow 0$,
    \begin{align*}
  &\large|      \sum_{k} \mathbb{E}[\partial_{t} b_{k} \beta_{k}^{0} ] - \partial_{t} b_{k}(t_{i},x_{i}) \beta_{k}^{0}  \\
   & -  \sum_{k} \hat{\mathbb{E}}[\partial_{t} b_{k}  ]\hat \beta_{k} - \partial_{t} b_{k}(t_{i},x_{i}) \hat \beta_{k} \large| = o_{P}(1).
    \end{align*}
    Using equation~\eqref{eq:24}, it remains to show that
 \begin{align*}
 \max_{i}  \left|  \frac{\epsilon_{i} \tilde{\mathbf{Z}}_{i}^{0}}{\mathbb{E}[\tilde{\mathbf{Z}}_{1}^{0}  \tilde{\mathbf{X}}_{1;1}^{0}]
      } -   \frac{\hat \epsilon_{i} \tilde{\mathbf{Z}}_{i}}{(\tilde{\mathbf{Z}})^{\intercal}  \tilde{\mathbf{X}}_{1}/n|} \right| = o_{P}(1)
 \end{align*}
Expanding the terms,
 \begin{align}\label{eq:25}
   \begin{split}
   &\left|  \frac{\epsilon_{i} \tilde{\mathbf{Z}}_{i}^{0}}{\mathbb{E}[\tilde{\mathbf{Z}}_{1}^{0}  \tilde{\mathbf{X}}_{1;1}^{0}]
   } -   \frac{\hat \epsilon_{i} \tilde{\mathbf{Z}}_{i}}{(\tilde{\mathbf{Z}})^{\intercal}  \tilde{\mathbf{X}}_{1}/n} \right| \\
   &\le   \left|  \frac{\epsilon_{i} \tilde{\mathbf{Z}}_{i}^{0} - \hat \epsilon_{i} \tilde{\mathbf{Z}}_{i}}{\mathbb{E}[\tilde{\mathbf{Z}}_{1}^{0}  \tilde{\mathbf{X}}_{1;1}^{0}]
      } \right|  + \left|   \hat \epsilon_{i} \tilde{\mathbf{Z}}_{i} \left(  \frac{1}{(\tilde{\mathbf{Z}})^{\intercal}  \tilde{\mathbf{X}}_{1}/n} -  \frac{1}{\mathbb{E}[\tilde{\mathbf{Z}}_{1}^{0} \tilde{\mathbf{X}}_{1;1}^{0}]}
     \right)  \right| \\
   &\le \left|  \frac{\epsilon_{i} \tilde{\mathbf{Z}}_{i}^{0} - \hat \epsilon_{i} \tilde{\mathbf{Z}}_{i}}{\mathbb{E}[\tilde{\mathbf{Z}}_{1}^{0}  \tilde{\mathbf{X}}_{1;1}^{0}]
     } \right| \\
   &+ \left|   \left( \hat \epsilon_{i} \tilde{\mathbf{Z}}_{i} - \epsilon_{i} \tilde{\mathbf{Z}}_{i}^{0}  \right)\left(  \frac{1}{(\tilde{\mathbf{Z}})^{\intercal}  \tilde{\mathbf{X}}_{1}/n} -  \frac{1}{\mathbb{E}[\tilde{\mathbf{Z}}_{1}^{0} \tilde{\mathbf{X}}_{1;1}^{0}]}
     \right)  \right| \\
   &+\left|    \epsilon_{i} \tilde{\mathbf{Z}}_{i}^{0}  \left(  \frac{1}{(\tilde{\mathbf{Z}})^{\intercal}  \tilde{\mathbf{X}}_{1}/n} -  \frac{1}{\mathbb{E}[\tilde{\mathbf{Z}}_{1}^{0} \tilde{\mathbf{X}}_{1;1}^{0}]}
     \right)  \right|
   \end{split}
 \end{align}
 We have shown in Proposition~\ref{prop-defw2} that
\begin{equation*}
  \tilde{\mathbf{Z}}^{\intercal} \tilde{\mathbf{X}}_{1}/n = \mathbb{E}[\tilde{\mathbf{Z}}_{1}^{0}  \tilde{\mathbf{X}}_{11}^{0}] + o_{P}(1),
\end{equation*}
and that the latter quantity is bounded away from zero. Furthermore, by assumption $  \max_{i} | \epsilon_{i} \tilde{\mathbf{Z}}_{i}^{0}|$ is bounded. Using equation~\eqref{eq:25} it suffices to show that
\begin{equation*}
 \max_i |\epsilon_{i} \tilde{\mathbf{Z}}_{i}^{0} - \hat \epsilon_{i} \tilde{\mathbf{Z}}_{i}| = o_{P}(1).
\end{equation*}
To this end note that
\begin{align*}
  \begin{split}
  & \max_{i}|\epsilon_{i} \tilde{\mathbf{Z}}_{i}^{0} - \hat \epsilon_{i} \tilde{\mathbf{Z}}_{i}| \\
  &\le \max_{i}|\epsilon_{i} (\tilde{\mathbf{Z}}_{i}^{0} - \tilde{\mathbf{Z}}_{i})| + \max_{i}| (\epsilon_{i} - \hat \epsilon_{i}) (\tilde{\mathbf{Z}}_{i}- \tilde{\mathbf{Z}}_{i}^{0})| \\
 & + \max_{i}| (\epsilon_{i} - \hat \epsilon_{i})  \tilde{\mathbf{Z}}_{i}^{0})|.
  \end{split}
\end{align*}
Now use the following inequalities
    \begin{align*}
      \begin{split}
      \|\tilde{\mathbf{Z}}^{0} \|_{\infty} &\le C_{3} < \infty \\
       \| \tilde{\mathbf{Z}} - \tilde{\mathbf{Z}}^{0}\|_{\infty} & = \| \tilde{\mathbf{X}}_{-1} \hat \gamma  - \tilde{\mathbf{X}}_{-1}^{0} \gamma^{0} \|_{\infty} \\
      &\le \| \tilde{\mathbf{X}}_{-1} (\hat \gamma - \gamma^{0}) \|_{\infty} + \| (\tilde{\mathbf{X}}_{-1} - \tilde{\mathbf{X}}_{-1}^{0})   \gamma^{0} \|_{\infty} \\
                                     &\le \|\tilde{\mathbf{X}}_{-1} \|_{\infty} \| \hat \gamma - \gamma^{0} \|_{1} + \| \tilde{\mathbf{X}}_{-1} - \tilde{\mathbf{X}}_{-1}^{0} \|_{\infty} \| \gamma^{0} \|_{1} \\
      &= o_{P}(1) \\
       \|\hat \epsilon - \epsilon \|_{\infty} 
      &\le  \| \tilde{\mathbf{X}} (\hat \beta - \beta^{0}) \|_{\infty} +  \| (\tilde{\mathbf{X}} - \tilde{\mathbf{X}}^{0}) \beta^{0} \|_{\infty}\\
                                     & \le \| \tilde{\mathbf{X}} \|_{\infty} \| \hat \beta - \beta^{0} \|_{1} + \| \tilde{\mathbf{X}} - \tilde{\mathbf{X}}^{0} \|_{\infty} \| \beta^{0} \|_{1} \\
                                     &= o_{P}(1)
     \end{split}
    \end{align*}
    Here we used that by a sub-Gaussian tail bound \citep[Chapter 2]{boucheron2013concentration}, (A2) implies $\| \tilde{\mathbf{X}} - \tilde{\mathbf{X}}^{0} \|_{\infty} = O_{P}(\sqrt{\log(p)/n})$. Furthermore, we used that by (A2), $\| \tilde{\mathbf{X}} \|_{\infty} = O_{P}(1)$, that by (D3) we have $\| \hat \beta - \beta^{0} \|_{1} = O_{P}(1/\sqrt{\log(p)})$, by (D2) we have $\|\hat \gamma - \gamma^{0}\|_{1} = O_{P}(1/\sqrt{\log(p)})$ and by assumption $\| \beta^{0} \|_{1} = O(1)$ and $\| \gamma^{0} \|_{1} = o(\sqrt{n/\log(p)})$. Hence, we have shown that
      \begin{align*}
    \frac{1}{n} \sum_{i} \xi_{i}^{0} - \xi_{i} &=  o_{P}(1), \\
    \frac{1}{n} \sum_{i}  (\xi_{i}^{0})^{2} - \xi_{i}^{2} &=  o_{P}(1).
      \end{align*}
As argued above, this concludes the proof.

\end{proof}

\end{proposition}

\subsubsection{Proof of Theorem~\ref{thm:doubly-robust-conf}}

\begin{proof}
  Combine Lemma~\ref{le:error} and Lemma~\ref{lem:D2D3} with Proposition~\ref{prop:almost}  and Proposition~\ref{prop:varianceestimation}. Note that due to assumption (A6), $u$ is bounded away from zero. Thus,
  \begin{equation*}
    \frac{\hat u^{2}}{u^{2}} = 1 + o_{P}(1),
  \end{equation*}
  which completes the proof.
\end{proof}

\subsection{Proof of Lemma~\ref{lem:trick}}

\begin{proof}
  Define $\overline b_{k} = b_{k} - t \mathbb{E}[\partial_{t} b_{k}]$ for $k>1$ and $\overline b_{1} = t$. By definition of $\beta^{0}$,
  \begin{align*}
    \beta^{0} &= \arg \min_{\beta} \mathbb{E}[\| \mathbf{Y} - \tilde{\mathbf{X}}^{0} \beta \|_{2}^{2}]  \\ &= \arg \min_{\beta} \mathbb{E}[ ( Y - \sum_{k} \overline b_{k}(T,X) \beta )^{2}].
  \end{align*}
  Now use Lemma~\ref{lem:doublyrobust}. This implies that
  \begin{equation*}
  \mathbb{E}[\partial_{t} \mathbb{E}[Y|X=x,T=t]] = \mathbb{E}[\sum_{k} \partial_{t} \overline b_{k} \beta_{k}^{0}].
\end{equation*}
Expanding the definition,
\begin{align*}
  & \, \,\mathbb{E}[\partial_{t} \mathbb{E}[Y|X=x,T=t]] \\
  &= \sum_{k}(\mathbb{E}[\partial_{t} b_{k}] - 1_{k > 1} \partial_{t} t \mathbb{E}[\partial_{t} b_{k}]) \beta_{k}^{0} \\
                                                 &= \sum_{k}(\mathbb{E}[\partial_{t} b_{k}] - 1_{k > 1} \mathbb{E}[\partial_{t} b_{k}]) \beta_{k}^{0} \\
                                                 &= \beta_{1}^{0}.
\end{align*}
This concludes the proof.
\end{proof}
\subsection{Proof of Lemma~\ref{lem:doublyrobust}}\label{sec:doublyrobust}

\begin{lemma}\label{lem:doublyrobust}
Define
\begin{equation*}
 b^{0} =   \arg \min_{b \in \mathcal{B}} \mathbb{E}[(f^{0}(T,X)- b(T,X))^2],
\end{equation*}
and
\begin{equation*}
 b_{*} = \arg \min_{b \in \mathcal{B}} \mathbb{E}[(f_*(T,X)- b(T,X))^2],
\end{equation*}
where $f_{*} = - \partial_t \log p(t|x)$. If, $\mathbb{P}$-a.s. we have
\begin{equation*}
    b^{0}(T,X) = f^{0}(T,X) \text{ or } b_{*}(T,X) = f_*(T,X),
\end{equation*}
then $\mathbb{E}[\partial_t f^{0} ] =\mathbb{E}[\partial_t b^{0}]$. %
\end{lemma}

\begin{proof}

As $\mathbb{P}$-a.s. we have $b^{0} = f^{0}$ or $b_{*} = f_*$, we also have that $\mathbb{P}$-a.s. $(b^{0} - f^{0} ) ( b_{*} - f_*) = 0$. Hence,
\begin{align*}
 0 = \, &\mathbb{E}[  (b^{0} - f^{0} ) ( b_{*} - f_*) ] \\
  =\, &  \mathbb{E}[f^{0} f_*]  + \mathbb{E}[b^{0} b_{*}] - \mathbb{E}[b^{0}f_{*}] -\mathbb{E}[f^{0}b_{*}] \\
  = \, &  \mathbb{E}[f^{0} f_*]  + \mathbb{E}[b^{0} b_{*}] -\mathbb{E}[b^{0} b_{*}] -\mathbb{E}[b^{0} b_{*}] \\
  = \, & \mathbb{E}[f^{0} f_*] - \mathbb{E}[b^{0} b_{*}] \\
  = \, & \mathbb{E}[f^{0} f_*] - \mathbb{E}[b^{0} f_{*}]
\end{align*}
Here, we used repeatedly that $\mathbb{E}[f^{0} b_{*}] = \mathbb{E}[b^{0} b_{*}]$ and that $\mathbb{E}[b^{0} f_{*}] = \mathbb{E}[b^{0} b_{*}]$. Now, using that 
\begin{equation*}
    0 = \mathbb{E}[f^{0} f_*] - \mathbb{E}[b^{0} f_{*}] =   \mathbb{E}[\partial_t f^{0} ] -\mathbb{E}[\partial_t b^{0}],
\end{equation*}
completes the proof.
\end{proof}

\subsection{Proof of Lemma~\ref{lem:semieffic}}

\begin{lemma}[Semiparametric efficiency bound]\label{lem:semieffic}
  Let the assumptions of Theorem~\ref{thm:doubly-robust-conf} hold. If $f^{0} \in \mathcal{B}$ and $\partial_{t} \log p(t|x) \in \mathcal{B}$, %
then the asymptotic variance of $\sqrt{n}(\hat \beta_{1}^{\text{despar}} - \beta_{1}^{0})$ is equal to $\text{Var}(\partial_{t} f^{0}) + \text{Var}(\epsilon \cdot \partial_{t} \log p(T|X))$. %
\end{lemma}
\begin{proof}
First, by Proposition~\ref{prop:varianceestimation}, the asymptotic variance of $\sqrt{n}(\hat \beta_{1}^{\text{despar}} - \beta_{1}^{0})$ is
  \begin{equation*}
    \frac{ \text{Var}(\epsilon_1 \tilde{\mathbf{Z}}_{1}^{0})}{\text{Var}(\tilde{\mathbf{Z}}_{1}^{0})^{2}} + \text{Var}(\partial_{t} f^{0})
  \end{equation*}
Thus, it suffices to show that
\begin{equation*}
    \mathbf{Z}_1^0 = \frac{f_*}{\mathbb{E}[f_*^2]},
\end{equation*}
 where $f_{*} = -\partial_{t} \log p(t_1,x_1)$. Let us first consider a univariate regression of $T$ on $ f_{*} - t \mathbb{E}[\partial_{t} f_{*}]$. Then, the residual variance is
 \begin{align*}
   &\min_{\alpha} \mathbb{E}[(T + \alpha(f_{*}-T \mathbb{E}[\partial_{t}f_{*}]))^{2}] \\
   =  &  \min_{\alpha} \mathbb{E}[(T(1- \alpha \mathbb{E}[\partial_{t}f_{*}]) + \alpha f_{*})^{2}]
 \end{align*}
 Expanding, and using that $\mathbb{E}[T f_{*}] = 1$,
 \begin{align*}
   & \, \, \mathbb{E}[(T(1- \alpha \mathbb{E}[\partial_{t}f_{*}]) + \alpha f_{*})^{2}] \\
   & = \mathbb{E}[T^{2}] (1 - 2 \alpha \mathbb{E}[\partial_{t}f_{*}] + \alpha^{2}  \mathbb{E}[\partial_{t}f_{*}]^{2} ) \\
   & \, \, + 2 \alpha (1- \alpha \mathbb{E}[\partial_{t}f_{*}]) + \alpha^{2} \mathbb{E}[f_{*}^{2}].
 \end{align*}
 Now we can use that $\mathbb{E}[f_{*}^{2}] = \mathbb{E}[\partial_{t}f_{*}]$. Taking the derivative with respect to $\alpha$, we obtain
 \begin{align*}
   -2\mathbb{E}[T^{2}] \mathbb{E}[f_{*}^{2}] + 2 \alpha \mathbb{E}[T^{2}] \mathbb{E}[f_{*}^{2}]^{2} + 2 - 4 \alpha \mathbb{E}[f_{*}^{2}] + 2 \alpha \mathbb{E}[f_{*}^{2}]
 \end{align*}
 Setting this term to zero and rearranging
 \begin{align*}
   -2\mathbb{E}[T^{2}] \mathbb{E}[f_{*}^{2}]  +2 &=     2 \alpha \mathbb{E}[f_{*}^{2}] - 2 \alpha \mathbb{E}[T^{2}] \mathbb{E}[f_{*}^{2}]^{2} \\
   &= \alpha \mathbb{E}[f_{*}^{2}] ( 2 - 2\mathbb{E}[f_{*}^{2}] \mathbb{E}[T^{2}])
 \end{align*}
 Thus, the solution is $\alpha = 1/\mathbb{E}[f_{*}^{2}]$ and the resulting residual variance is
  \begin{equation*}
   \min_{\alpha} \mathbb{E}[(T + \alpha(f_{*}-T \mathbb{E}[\partial_{t}f_{*}]))^{2}] =    \mathbb{E}[(f_{*}/\mathbb{E}[f_{*}^{2}])^{2}].
 \end{equation*}
 By definition, $f_{*}$ is uncorrelated with $\tilde b_{k}$ for all $k>1$. Thus,
 \begin{equation*}
   \min_{\alpha_{1},\ldots,\alpha_{p}} \mathbb{E}[(T - \sum_{k>1} \alpha_{k}\tilde b_{k})^{2}]= \mathbb{E}[(f_{*}/\mathbb{E}[f_{*}^{2}])^{2}].
 \end{equation*}
 As the minimizer is unique, $\mathbf{Z}_1^0=\frac{f_*}{\mathbb{E}[f_*^2]}$.
\end{proof}

\bibliography{references}

\end{document}